\documentclass[11pt]{article}
\usepackage{amsfonts,amsmath,amssymb,amsthm,enumerate,hyperref,color}

\newtheorem{observation}{Observation}
\newtheorem{theorem}{Theorem}
\newtheorem{fact}{Fact}
\newtheorem{lemma}{Lemma}

\newtheorem{claim}{Claim}
\newtheorem{definition}{Definition}
\newtheorem{corollary}{Corollary}
\newtheorem{question}{Question}

\newcommand{\ignore}[1]{}

\newtheorem{proposition}       [theorem]   {Proposition}

\def\E{\mathrm{E}}

\def\cA{\mathcal{A}}

\def\cF{\mathcal{F}}

\def\cO{\mathcal{O}}

\def\cS{\mathcal{S}}

\def\01{\{0,1\}}
\def\r01{[0,1]}
\def\reals{{\bf R}}

\def\eps{\epsilon}
\def\opt{\mathrm{opt}}

\def\poly{\mathrm{poly}}

\newtheorem{Alg}{Algorithm}

\title{Dueling algorithms}
\author{
Nicole Immorlica\thanks{Department of Electrical Engineering and Computer Science, Northwestern University} \thanks{Part of this work was performed while the author was at Microsoft Research}
\and Adam Tauman Kalai\thanks{Microsoft Research New England}
\and Brendan Lucier\thanks{Department of Computer Science, University of Toronto} $^\dag$
\and Ankur Moitra\thanks{Department of Electrical Engineering and Computer Science, Massachusetts Institute of Technology.  Supported in part by a Fannie Hurts Fellowship.} $\,^\dag$
\and Andrew Postlewaite\thanks{Department of Economics, University of Pennsylvania} $^\dag$
\and Moshe Tennenholtz\thanks{Microsoft R\&D Israel and the Technion, Israel}}

\begin{document}

\setcounter{page}{0}

\maketitle

\abstract{ We revisit classic algorithmic search and optimization
  problems from the perspective of competition.  Rather than a single
  optimizer minimizing expected cost, we consider a zero-sum game in
  which an optimization problem is presented to two players, whose only goal
  is to {\em outperform the opponent}.  Such games are typically
  exponentially large zero-sum games, but they often have a rich combinatorial
  structure.  We provide general techniques by which such structure
  can be leveraged to find minmax-optimal and approximate
  minmax-optimal strategies.  We give examples of ranking,
  hiring, compression, and binary search duels, among others.  We give
  bounds on how often one can beat the classic optimization algorithms
  in such duels.}

\newpage

\section{Introduction}

Many natural optimization problems have two-player competitive
analogs.  For example, consider the ranking problem of selecting an
order on $n$ items, where the cost of searching for a single item is
its rank in the list.  Given a fixed probability distribution over
desired items, the trivial greedy algorithm, which orders items in
decreasing probability, is optimal.

Next consider the following natural two-player version of the problem,
which models a user choosing between two search engines.  The user
thinks of a desired web page and a query and executes the query on
both search engines.  The engine that ranks the desired page higher is
chosen by the user as the ``winner.''  If the greedy algorithm has
the ranking of pages $\omega_1,\omega_2,\ldots,\omega_n$, then the
ranking $\omega_2,\omega_3,\ldots,\omega_n,\omega_1$ beats the greedy
ranking on every item 
except $\omega_1$.  We say
the greedy algorithm is $1-1/n$ {\em beatable} because there is a
probability distribution over pages for which the greedy algorithm
loses $1-1/n$ of the time.
Thus, in a competitive
setting, an ``optimal'' search engine can perform poorly against
a clever opponent.

This {\em ranking duel} can be modeled as a symmetric constant-sum
game, with $n!$ strategies, in which the player with the higher
ranking of the target page receives a payoff of 1 and the other
receives a payoff of 0 (in the case of a tie, say they both receive a
payoff of 1/2).  As in all symmetric one-sum games, there must be
(mixed) strategies that guarantee expected payoff of at least 1/2
against any opponent.  Put another way, there must be a (randomized)
algorithm that takes as input the probability distribution and outputs
a ranking, which is guaranteed to achieve expected payoff of at least
$1/2$ against any opposing algorithm.

This conversion can be applied to any optimization
problem with an element of uncertainty.  Such problems are of the form
$\min_{x \in X} \E_{\omega\sim p}[c(x,\omega)]$, where $p$ is a
probability distribution over the {\em state of nature} $\omega \in
\Omega$, $X$ is a feasible set, and $c:X\times\Omega\rightarrow
\reals$ is an objective function.  The dueling analog has two players
simultaneously choose $x,x'$; player 1 receives payoff 1 if
$c(x,\omega)<c(x',\omega)$, payoff $0$ if
$c(x,\omega)>c(x',\omega)$, payoff $1/2$ otherwise, and similarly
for player 2.\footnote{Our techniques will also apply to asymmetric payoff
  functions; see Appendix \ref{app:asymmetric}.}

There are many natural examples of this setting beyond the ranking
duel mentioned above.  For example, for the shortest-path routing
under a distribution over edge times, the corresponding {\it racing
  duel} is simply a race, and the state of nature encodes uncertain
edge delays.\footnote{
We also refer to this as the {\em primal duel} because
any other duel can be represented as a race with an appropriate
graph and probability distribution $p$, though there may be an
exponential blowup in representation size.}
For the classic secretary problem, in the corresponding {\it hiring
  duel} two employers must each select a candidate from a pool of $n$
candidates (though, as standard, they must decide whether or not to
choose a candidate before interviewing the next one), and the winner
is the one that hires the better candidate.  This could model, for
example, two competing companies attempting to hire CEOs or two
opposing political parties selecting politicians to run in an
election; the absolute quality of the candidate may be less important
than being better than the other's selection.
%
%
In a {\it compression duel}, a user with a (randomly chosen) sample
string $\omega$ chooses between two compression schemes based on which
one compresses that string better.  This setting can also model a user searching
for a file in two competing, hierarchical storage systems and choosing the system
that finds the file first. In a {\it binary search duel}, a
user searches for a random element in a list using two different
search trees, and chooses whichever tree finds the element faster.

\paragraph{Our contribution.}
For each of these problems, we consider a number of questions related to how vulnerable
a classic algorithm is to competition, what algorithms will be selected at equilibrium, and how well these strategies at equilibrium solve the original optimization problem.
\begin{question}
Will players use the classic optimization solution in the dueling setting?
\end{question}
Intuitively, the answer to this question should depend on how much an opponent can {\it game} the classic optimization solution. For example, in the {\it ranking duel} an opponent can beat the greedy algorithm on almost all pages -- and even the most oblivious player would quickly realize the need to change strategies. In contrast, we demonstrate that many classic optimization solutions  -- such as the secretary algorithm for hiring, Huffman coding for compression, and standard binary search -- are substantially less vulnerable. We say an algorithm is $\beta$-beatable (over
distribution $p$) if there exists a response which achieves payoff
$\beta$ against that algorithm (over distribution $p$). We summarize our results on the beatability of the standard optimization algorithm in each of our example optimization problems in the table below:

\begin{center}
\begin{tabular}{| l | c | c | }
\hline
Optimization Problem & Upper Bound & Lower Bound \\
\hline
Ranking & $1 - 1/n$ & $1 - 1/n$ \\
Racing & $1$ & $1$ \\
Hiring & $0.82$ & $0.51$ \\
Compression & $3/4$ & $2/3$ \\
Search & $5/8$ & $5/8$ \\
\hline
\end{tabular}
\end{center}

\begin{question}
What strategies do players play at equilibrium?
\end{question}

We say an algorithm efficiently {\em solves} the duel if it takes as input a
representation of the game and probability distribution $p$, and
outputs an action $x \in X$ distributed according to some minmax
optimal (i.e., Nash equilibrium) strategy.  As our main result, we
give a general method for solving duels that can be represented in a
certain bilinear form.  We also show how to convert an approximate
best-response oracle for a dueling game into an approximate minmax
optimal algorithm, using techniques from low-regret learning.  We
demonstrate the generality of these methods by showing how to apply
them to the numerous examples described above. For many problems we
consider, the problem of computing minmax optimal strategies reduces to
finding a simple description of the space of feasible mixed strategies 
(i.e. expressing this set as the projection of a polytope with polynomially many
variables and constraints). See \cite{Yann} for a thorough treatment of such problems.

\begin{question}
Are these equilibrium strategies still good at solving the optimization problem?
\end{question}

As an example, consider the ranking duel. How much more time does a web surfer need
to spend browsing to find the page he is interested in, because more than one search
engine is competing for his attention? In fact, the surfer may be \emph{better} off due
to competition, depending on the model of comparison. For example, the cost to the web
surfer may be the minimum of the ranks assigned by each search engine. And we leave
open the tantalizing possibility that this quantity could in general be smaller at equilibrium for two
competing search engines than for just one search engine playing the greedy algorithm.


\vspace{-2mm}
\paragraph{Related work.}
The work most relevant to ours is the study of ranking
games~\cite{BFHS09}, and more generally the study of social context
games~\cite{AKT08}.
In these settings, players' payoffs are translated
into utilities based on social contexts, defined by a graph and an
aggregation function. For example, a player's utility can be the
sum/max/min of his neighbors' payoffs.  This work studies the effect
of social contexts on the existence and computation of game-theoretic
solution concepts,
%
but does not re-visit optimization algorithms in competitive settings.

For the hiring problem, several competitive variants and their
algorithmic implications have been considered (see, e.g.,~\cite{IKM06}
and the references therein).  A typical competitive setting is a
(general sum) game where a player achieves payoff of 1 if she hires
the very best applicant and zero otherwise.  But, to the best of our
knowledge, no one has considered the natural model of a duel where
the objective is simply to hire a better candidate than the
opponent.
Also related to our algorithmic results are succinct zero-sum games,
where a game has exponentially many strategies but the payoff function
can be computed by a succinct circuit. This general class has been
showed to be EXP-hard to solve~\cite{FKS95}, and also difficult to
approximate~\cite{FIKU05}.

Finally, we note the line of research on competition among mechanisms,
such as the study of competing auctions (see
e.g. \cite{BS99,Mcafee93,MT04,PS97}) or schedulers \cite{ATZ10}.
In such settings, each player selects a mechanism and then bidders
select the auction to
participate in and how much to bid there, where both designers and
bidders are strategic. This work is largely concerned with the existence
of sub-game perfect equilibrium.

\vspace{-2mm}
\paragraph{Outline.}
In Section \ref{sec:defn} we define our model formally and provide a
general framework for solving dueling problems as well as the warmup
example of the ranking duel.  We then use these tools to analyze the
more intricate settings of the hiring duel (Section~\ref{sec:hiring}),
the compression duel (Section~\ref{sec:compression}), and the search
duel (Section~\ref{sec:bst}).  We describe
avenues of future research in Section~\ref{sec:conc}.

\section{Preliminaries}
\label{sec:defn}

A problem of optimization under uncertainty, $(X,\Omega,c,p)$, is
specified by a feasible set $X$, a commonly-known distribution $p$
over the state of nature, $\omega$, chosen from set $\Omega$,
and an objective function $c:X \times \Omega\rightarrow \reals$.  For
simplicity we assume all these sets are finite.  When $p$ is clear
from context, we write the expected cost of $x \in X$ as $c(x) =
\E_{\omega \sim p}[c(x,\omega)]$. The one-player optimum is
$\opt=\min_{x \in X} c(x)$.  Algorithm $A$ takes as input $p$ and
randomness $r \in [0,1]$, and outputs $x\in X$.  We define
$c(A)=\E_r[c(A(p,r))]$ and an algorithm $A$ is {\em one-player
  optimal} if $c(A)=\opt$.

In the two-person constant-sum duel game $D(X,\Omega,c,p)$, players
simultaneously choose $x,x'\in X$, and player 1's payoff is:
$$v(x,x',p)=\Pr_{\omega \sim p}[c(x,\omega)<c(x',\omega)]+
\frac{1}{2}\Pr_{\omega \sim p}[c(x,\omega)=c(x',\omega)].$$
When $p$ is understood from context we write $v(x,x')$.  Player 2's
payoff is $v(x',x)=1-v(x,x')$.  This models a tie,
$c(x,\omega)=c(x',\omega)$, as a half point for each.  We define the
value of a strategy, $v(x,p)$, to be how much that strategy
guarantees,
$v(x,p)=\min_{x'\in X} v(x,x',p).$
Again, when $p$ is understood from context we write simply $v(x)$.

The set of probability distributions over set $S$ is denoted
$\Delta(S)$.  A {\em mixed strategy} is $\sigma \in \Delta(X)$.
As is standard, we extend the domain of $v$ to mixed strategies
bilinearly by expectation.  A {\em best response} to mixed strategy
$\sigma$ is a strategy which yields maximal payoff against
$\sigma$, i.e., $\sigma'$ is a best response to $\sigma$ if it
maximizes $v(\sigma',\sigma).$ A {\em minmax} strategy is a (possibly
mixed) strategy that guarantees the safety value, in this case
1/2, against any opponent play.  The best response to
such a strategy yields payoffs of 1/2.  The set of minmax strategies
is denoted $MM(D(X,\Omega,c,p))=\{\sigma \in
\Delta(X)~|~v(\sigma)=1/2\}$.  A basic fact about constant-sum games
is that the set of Nash equilibria is the cross product of the minmax
strategies for player 1 and those of player 2.

\subsection{Bilinear duels}\label{sec:bilinear}

In a bilinear duel, the feasible set of strategies are points in
$n$-dimensional Euclidean space, i.e., $X \subseteq \reals^n$, $X'
\subseteq \reals^{n'}$ and the payoff to player 1 is
$v(x,x')=x^t M x'$
for some matrix $M \in \reals^{n \times n'}$.  In $n\times n$ bimatrix
games, $X$ and $X'$ are just simplices $\{x\in \reals_{\geq 0}^n ~|~ \sum x_i
=1\}$.  Let $K$ be the convex hull of $X$.  Any point in $K$ is
achievable (in expectation) as a mixed strategy.  Similarly define
$K'$.  As we will point out in this section, solving these reduces to
linear programming with a number of constraints proportional to the
number of constraints necessary to define the feasible sets, $K$ and
$K'$.  (In typical applications, $K$ and $K'$ have a polynomial number
of facets but an exponential number of vertices.)

Let $K$ be a polytope defined by the intersection of $m$ halfspaces,
$K = \{x \in \reals^n ~|~w_i \cdot x \geq b_i \text{ for
}i=1,2,\ldots,m\}.$ Similarly, let $K'$ be the intersection of $m'$
halfspaces $w_i' \cdot x \geq b_i'$.  The typical way to reduce to an
LP for constant-sum games is:
$$\max_{v \in \reals, x\in \reals^n} v \text{ such that } x \in K \text{ and } x^TMx' \geq v \text{ for all }x'\in X'.$$
The above program has a number of constraints which is $m+|X'|$, ($m$
constraints guaranteeing that $x\in K$), and $|X'|$ is typically
exponential.  Instead, the following linear program has $O(n'+m+m')$
constraints, and hence can be found in time polynomial in $n',m,m'$
and the bit-size representation of $M$ and the constraints in $K$ and
$K'$.
\begin{equation}\label{eq:LP}
\max_{x\in \reals^n, \lambda \in \reals^{m'}} \sum_1^{m'} \lambda_i b_i' \text{ such that } x \in K \text{ and } x^tM = \sum_1^{m'} \lambda_i w_i'.
\end{equation}
\begin{lemma}
\label{lem:lp}
For any constant-sum game with strategies $x \in K, x'\in K$ and
payoffs $x^t M x'$, the maximum of the above linear program is the
value of the game to player 1, and any maximizing $x$ is a minmax
optimal strategy.
\end{lemma}

\begin{proof}
First we argue that the value of the above LP is at least as large as
the value of the game to player 1.  Let $x,\lambda$ maximize the above
LP and let the maximum be $\alpha$.  For any $x' \in K'$, $$x^t M x' =
\sum_1^{m'} \lambda_i w_i' \cdot x' \geq \sum_1^{m'} \lambda_i b_i' =
\alpha.$$ Hence, this means that strategy $x$ guarantees player $x$ at
least $\alpha$ against any opponent response, $x' \in K$.  Hence
$\alpha \leq v$ with equality iff $x$ is minmax optimal.  Next, let
$x$ be any minmax optimal strategy, and let $v$ be the value of the
constant-sum game.  This means that $x^t M x' \geq v$ for all $x' \in
K'$ with equality for some point.  In particular, the minmax theorem
(equivalently, duality) means that the LP $\min_{x' \in K'} x^t M x'$
has a minimum value of $v$ and that there is a vector of $\lambda\geq
0$ such that $\sum_1^{m'} \lambda_i w_i' = x^t M$ and $\sum_1^{m'}
\lambda_i b_i'=v$.  Hence $\alpha \geq v$.
\end{proof}

\subsection{Reduction to bilinear duels}\label{sec:blexact}
The sets $X$ in a duel are typically objects such as paths, trees,
rankings, etc., which are not themselves points in Euclidean space.
In order to use the above approach to reduce a given duel
$D(X,\Omega,c,p)$ to a bilinear duel in a {\em computationally
  efficient manner}, one needs the following:
\begin{enumerate}
\item An efficiently computable function $\phi:X \rightarrow K$ which
  maps any $x \in X$ to a feasible point in $K \subseteq \reals^n$.
\item A payoff matrix $M$ demonstrating such that $v(x,x') = \phi(x)^t
  M \phi(x')$, demonstrating that the problem is indeed bilinear.
\item A set of polynomially many feasible constraints which defines
  $K$.
\item A ``randomized rounding algorithm'' which takes as input a point
  in $K$ outputs an object in $X$.
\end{enumerate}
In many cases, parts (1) and (2) are straightforward.  Parts (3) and
(4) may be more challenging.  For example, for the binary trees used
in the compression duel, it is easy to map a tree to a vector of node
depths.  However, we do not know how to efficiently determine whether
a given vector of node depths is indeed a mixture over trees (except
for certain types of trees which are in sorted order, like the binary
search trees in the binary search duel).  In the next subsection, we show how
computing approximate best responses suffices.

\subsection{Approximating best responses and approximating minmax}\label{sec:appx}
In some cases, the polytope $K$ may have exponentially or infinitely
many facets, in which case the above linear program is not very
useful. In this section, we show that if one can compute {\em
  approximate} best responses for a bilinear duel, then one can {\em
  approximate} minmax strategies.

For any $\eps>0$, an $\eps$-best response to a player 2 strategy $x'
\in K'$ is any $x \in K$ such that $x^t M x' \geq \min_{y \in K} y^T M
x' - \eps$.  Similarly for player 1.  An $\eps$-minmax strategy $x \in
K$ for player 1 is one that guarantees player 1 an expected payoff not
worse than $\eps$ minus the value, i.e.,
$$\min_{x' \in K} v(x,x')\geq \max_{y \in K}\min_{x' \in K} v(y,x') -
\eps.$$

Best response oracles are functions from $K$ to $K'$ and vice versa.  However, for many applications (and in particular the ones in this paper) where all feasible points are nonnegative, one can define a best response oracle for all nonnegative points in the positive orthant.  (With additional effort, one can remove this assumption using Kleinberg and Awerbuch's elegant notion of a Barycentric spanner \cite{AK04}.)  For scaling purposes, we assume that for some $B>0$, the convex sets are $K \subseteq [0,B]^n$ and $K' \subseteq [0,B]^{n'}$ and the matrix $M \in [-B,B]^{n \times n'}$ is bounded as well.

Fix any $\eps>0$.  We suppose that we are given an $\eps$-approximate best response oracle in the following sense.
For player 1, this is an oracle $\cO:[0,B]^{n'} \rightarrow K$ which has the property that
$\cO(x')^t M x'  \geq \max_{x \in K} x^t M x' -\eps$ for any $x'\in [0,B]^{n'}$.  Similarly for $\cO'$ for player 2.  Hence, one is able to potentially respond to things which are not feasible strategies of the opponent.  As can be seen in a number of applications, this does not impose a significant additional burden.

\begin{lemma}\label{lem:appx}
For any $\eps>0$, $n,n' \geq 1$, $B>0$, and any bilinear dual with convex $K \subseteq [0,B]^n$ and $K' \subseteq [0,B]^{n'}$ and $M \in [-B,B]^{n \times n'}$, and any $\eps$-best response oracles, there is an algorithm for
finding $\bigl(24(\eps\max(m,m'))^{1/3}B^2(nn')^{2/3}\bigr)$-minmax
strategies $x \in K, x'\in K'$.  The algorithm uses
$\poly(\beta,m,m',1/\eps)$ runtime and make $\poly(\beta,m,m',1/\eps)$
oracle calls.
\end{lemma}

The reduction and proof is deferred to Appendix~\ref{app:defn}.  It
uses Hannan-type of algorithms, namely ``Follow the expected leader''
\cite{KV05}.

We reduce the compression duel, where the base objects are trees, to a
bilinear duel and use the approximate best response oracle.
To perform such a reduction, one needs the following.
\begin{enumerate}
\item An efficiently computable function $\phi:X \rightarrow K$ which
  maps any $x \in X$ to a feasible point in $K \subseteq \reals^n$.
\item A bounded payoff matrix $M$ demonstrating such that $v(x,x') = \phi(x)^t
  M \phi(x')$, demonstrating that the problem is indeed bilinear.
\item $\eps$-best response oracles for players 1 and 2.  Here, the
  input to an $\eps$ best response oracle for player 1 is $x' \in [0,B]^{n'}$.
\end{enumerate}

\subsection{Beatability}
One interesting quantity to examine is how well a one-player
optimization algorithm performs in the two-player game.  In other
words, if a single player was a monopolist solving the one-player
optimization problem, how badly could they be beaten if a second
player suddenly entered.  For a particular one-player-optimal
algorithm $A$, we define its {\em beatability over distribution $p$}
to be $\E_r[v(A(p,r),p)]$, and we define its {\em beatability} to be
$\inf_p \E_r[v(A(p,r),p)]$.

\subsection{A warmup: the ranking duel}\label{sec:rank}
In the ranking duel, $\Omega = [n]=\{1,2,\ldots,n\}$, $X$ is the set
of permutations over $n$ items, and $c(\pi,\omega) \in [n]$ is the
position of $\omega$ in $\pi$ (rank 1 is the ``best'' rank).  The
greedy algorithm, which outputs permutation
$(\omega_1,\omega_2,\ldots,\omega_n)$ such that $p(\omega_1)\geq
p(\omega_2) \geq \cdots \geq p(\omega_n)$, is optimal in the
one-player version of the problem.\footnote{In some cases, such as a
  model of competing search engines, one could have the agents rank
  only $k$ items, but the algorithmic results would be similar.}
%

This game can be represented as a bilinear duel as follows.  Let $K$
and $K'$ be the set of doubly stochastic matrices, $K=K'=\{x \in
\reals_{\geq 0}^{n^2}~|~\forall j \sum_i x_{ij}=1, \forall i \sum_j
x_{ij}=1\}.$ Here $x_{ij}$ indicates the {\em probability} that item
$i$ is placed in position $j$, in some distribution over rankings.
The Birkhoff-von Neumann Theorem states that the set $K$ 
is precisely the set of probability distributions over rankings (where 
each ranking is represented as a permutation matrix  $x \in 
\{0,1\}^{n^2}$), and moreover any such $x \in K$ can be 
implemented efficiently via a form of randomized rounding.
See, for example, Corollary 1.4.15 of \cite{LP86}.
Note $K$ is a polytope in $n^2$ dimensions with $O(n)$ facets.  In
this representation, the expected payoff of $x$ versus $x'$ is
$$\sum_i p(i)\left(\frac{1}{2}\Pr[\text{Equally rank } i]+\Pr[\text{P1 ranks } i \text{ higher}]\right)=\sum_i p(i) \sum_j x_{ij} \left(\frac{1}{2}x'_{ij}+ \sum_{k>j} x'_{ik}\right).$$
The above is clearly bilinear in $x$ and $x'$ and can be written as
$x^tMx'$ for some matrix $M$ with bounded coefficients.  Hence, we can
solve the bilinear duel by the linear program (\ref{eq:LP}) and round
it to a (randomized) minmax optimal algorithm for ranking.

We next examine the beatability of the greedy algorithm.  Note that
for the uniform probability distribution $p(1)=p(2)=\ldots=p(n)=1/n$,
the greedy algorithm outputting, say, $(1,2,\ldots,n)$ can be beaten
with probability $1-1/n$ by the strategy $(2,3,\ldots,n,1)$.  One can
make greedy's selection unique by setting $p(i) =
1/n+(i-n/2)\epsilon$, and for sufficient small $\epsilon$ greedy can
be beaten a fraction of time arbitrarily close to $1-1/n$.

\section{Hiring Duel}\label{sec:hiring}



In a hiring duel, there are two employers $A$ and $B$ and two
corresponding sets of workers $U_A=\{a_1,\ldots,a_n\}$ and
$U_B=\{b_1,\ldots,b_n\}$ with $n$ workers each.  The $i$'th worker of
each set has a common value $v(i)$ where $v(i)>v(j)$ for all $i$ and
$j>i$.  Thus there is a total ranking of workers $a_i\in U_A$
(similarly $b_i\in U_B$) where a rank of $1$ indicates the best
worker, and workers are labeled according to rank.  The goal of the
employers is to hire a worker whose value (equivalently rank) beats
that of his competitor's worker.  Workers are interviewed by employers
one-by-one in a random order.  The relative ranks of workers are
revealed to employers only at the time of the interview.  That is, at
time $i$, each employer has seen a prefix of the interview order
consisting of $i$ of workers and knows only the projection of the
total ranking on this prefix.\footnote{In some cases, an employer also
  knows when and whom his opponent hired, and may condition his
  strategy on this information as well.  Only one of the settings
  described below needs this knowledge set; hence we defer our
  discussion of this point for now and explicitly mention the
  necessary assumptions where appropriate.}  Hiring decisions must be
made at the time of the interview, and only one worker may be hired.
Thus the employers' pure strategies are mappings from any prefix and
permutation of workers' ranks in that prefix to a binary hiring
decision.  We note that the permutation of ranks in a prefix does not
effect the distribution of the rank of the just-interviewed worker, and hence
without loss of generality we may assume the strategies are mapings
from the round number and current rank to a hiring decision.

In dueling notation, our game is $(X,\Omega,c,p)$ where the elements
of $X$ are functions $h:\{1,\ldots,n\}^2\rightarrow\{0,1\}$ indicating
for any round $i$ and projected rank of current interviewee $j\leq i$
the hiring decision $h(i,j)$; $\Omega$ is the set
$(\sigma_A,\sigma_B)$ of all pairs of permutations of $U_A$ and $U_B$;
$c(h,\sigma)$ is the value $v(\sigma^{-1}(i^*))$ of the first
candidate $i^*=\mathrm{argmin}_{i}\{i:h(i,[\sigma^{-1}(i)]_i)=1\}$ (where
$[\sigma^{-1}(i)]_j$ indicates the projected rank of the $i$'th
candidate among the first $j$ candidates according to $\sigma$) that
received an offer; and $p$ (as is typical in the secretary problem) is
the uniform distribution over $\Omega$.  The mixed strategies
$\pi\in\Delta(X)$ are simply mappings
$\pi:\{0,\ldots,n\}^2\rightarrow[0,1]$ from rounds and projected ranks
to a probability $\pi(i,j)$ of a hiring decision.

The values $v(\cdot)$ may be chosen adversarially, and hence in the
one-player setting the optimal algorithm against a worst-case
$v(\cdot)$ is the one that maximizes the probability of hiring the
best worker (the worst-case values set $v(1)=1$ and $v(i)<<1$ for
$i>1$).  In the literature on secretary problems, the following {\it
  classical algorithm} is known to hire the best worker with
probability approaching $\frac{1}{e}$: Interview n/e workers and hire
next one that beats all the previous.  Furthermore, there is no other
algorithm that hires the best worker with higher probability.
%

\subsection{Common pools of workers}
\label{subsec:commonhiring}

In this section, we study the {\it common hiring duel} in which
employers see the {\it same} candidates in the {\it same} order so
that $\sigma_A=\sigma_B$ and each employer observes when the other
hires.  In this case, the following strategy $\pi$ is a symmetric
equilibrium: If the opponent has already hired, then hire anyone who
beats his employee; otherwise hire as soon as the current candidate
has at least a $50\%$ chance of being the best of the remaining
candidates.

\begin{lemma}
\label{lem:commonequil}
Strategy $\pi$ is efficiently computable and constitutes a symmetric
equilibrium of the common hiring duel.
\end{lemma}

The computability follows from a derivation of probabilities in terms
of binomials, and the equilibrium claim follows by observing that
there can be no profitable deviation.  This strategy also beats the
classical algorithm, enabling us to provide non-trivial lower and
upper bounds for its beatability.

\begin{proof}
For a round $i$, we compute a threshold $t_i$ such that $\pi$ hires if
and only if the projected rank of the current candidate $j$ is at most
$t_i$.  Note that if $i$ candidates are observed, the probability that
the $t_i$'th best among them is better than all remaining candidates
is precisely ${i \choose t_i}/{n \choose t_i}$.  The numerator is the
number of ways to place the $1$ through $t_i$'th best candidates
overall among the first $i$ and the denominator is the number of ways
to place the $1$ through $t_i$'th best among the whole order.  Hence
to efficiently compute $\pi$ we just need to compute $t_i$ or,
equivalently, estimate these ratios of binomials and hire whenever on
round $i$ and observing the $j$'th best so far, ${i\choose
  j}/{n\choose j}\geq1/2$.


We further note $\pi$ is a symmetric equilibrium since if an employer
deviates and hires early then by definition the opponent has a better
than $50\%$ chance of getting a better candidate.  Similarly, if an
employer deviates and hires late then by definition his candidate has
at most a $50\%$ chance of being a better candidate than that of his
opponent.
\end{proof}

\begin{lemma}
\label{lem:hiringbeatability}
The beatability of the classical algorithm is at least $0.51$ and at
most $0.82$.
\end{lemma}

  The lower bound follows from the
fact that $\pi$ beats the classical algorithm with probability bounded
above $1/2$ when the classical algorithm hires early (i.e., before
round $n/2$), and the upper bound follows from the fact that the
classical algorithm guarantees a probability of $1/e$ of hiring the
best candidate, in which case no algorithm can beat it.

\begin{proof}
For the lower bound, note that in any event, $\pi$ guarantees a payoff
of at least $1/2$ against the classical algorithm.  We next argue that
for a constant fraction of the probability space, $\pi$ guarantees a
payoff of strictly better than $1/2$.  In particular, for some
$q,1/e<q<1/2$, consider the event that the classical algorithm hires
in the interval $\{n/e,qn\}$.  This event happens whenever the best
among the first $qn$ candidates is not among the first $n/e$
candidates, and hence has a probability of $(1-1/qe)$.  Conditioned on
this event, $\pi$ beats the classical algorithm whenever the best
candidate overall is in the last $n(1-q)$ candidates,\footnote{This is
  a loose lower bound; there are many other instances where $\pi$ also
  wins, e.g., if the second-best candidate is in the last $n(1-q)$
  candidates and the best occurs after the third best in the first
  $qn$ candidates.} which happens with probability $(1-q)$ (the
conditioning does not change this probability since it is only a
property of the permutation projected onto the first $qn$ elements).
Hence the overall payoff of $\pi$ against the classical algorithm is
$(1-q)(1-1/qe) + (1/2)(1/qe)$.  Optimizing for $q$ yields the result.

For the upper bound, note as mentioned above that the classical
algorithm has a probability approaching $1/e$ of hiring the {\it best}
candidate.  From here, we see $((1/2e) + (1-1/e))=1-1/2e<0.82$ is an
upper bound on the beatability of the classical algorithm since the
best an opponent can do is always hire the best worker when the
classical algorithm hires the best worker and always hire a better
worker when the classical algorithm does not hire the best worker.
\end{proof}

\subsection{Independent pools of workers}
\label{subsec:separatehiring}

In this section, we study the {\it independent hiring duel} in which
the employers see {\it different} candidates.  Thus
$\sigma_A\not=\sigma_B$ and the employers do not see when the opponent
hires.  We use the
bilinear duel framework introduced in Section~\ref{sec:bilinear} to
compute an equilibrium for this setting, yielding the following
theorem.

\begin{theorem}
\label{thm:separatehiring}
The equilibrium strategies of the independent hiring duel are
efficiently computable.
\end{theorem}

The main idea is to represent strategies $\pi$ by vectors $\{p_{ij}\}$
where $p_{ij}$ is the (total) probability of hiring the $j$'th best
candidate seen so far on round $i$.  Let $q_i$ be the probability of
reaching round $i$, and note it can be computed from the $\{p_{ij}\}$.
Recall $\pi(i,j)$ is the probability of hiring the $j$'th best so far
at round $i$ conditional on seeing the $j$'th best so far at round
$i$.  Thus using Bayes' Rule we can derive an efficiently-computable
bijective mapping (with an efficiently computable inverse) $\phi(\pi)$
between $\pi$ and $\{p_{ij}\}$ which simply sets
$\pi(i,j)=p_{ij}/(q_i/i)$.  It only remains to show that one can find
a matrix $M$ such that the payoff of a strategy $\pi$ versus a
strategy $\pi'$ is $\phi(\pi)^tM\phi(\pi')$.  This is done by
calculating the appropriate binomials.  

We show how to apply the bilinear duel framework to compute the
equilibrium of the independent hiring duel.  This requires the
following steps: define a subset $K$ of Euclidean space to represent
strategies, define a bijective mapping between $K$ and feasible
(mixed) strategies $\Delta(X)$, and show how to represent the payoff
matrix of strategies in the bilinear duel space.  We discuss each step
in order.

{\bf Defining $K$.}  For each $1\leq i\leq n$ and $j\leq i$ we
define $p_{ij}$ to be the (total) probability of seeing and hiring the
$j$'th best candidate seen so far at round $i$.  Our subspace
$K=[0,1]^{n(n+1)/2}$ consists of the collection of probabilities
$\{p_{ij}\}$.  To derive constraints on this space, we introduce a new
variable $q_i$ representing the probability of reaching round $i$.  We
note that the probability of reaching round $(i+1)$ must equal the
probability of reaching round $i$ and {\it not} hiring, so that
$q_{i+1}=q_i-\sum_{j=1}^np_{ij}$.  Furthermore, the probability
$p_{ij}$ can not exceed the probability of reaching round $i$ and
interviewing the $j$'th best candidate seen so far.  The probability
of reaching round $i$ is $q_i$ by definition, and the probability that
the projected rank of the $i$'th candidate is $j$ is $1/i$ by our
choice of a uniformly random permutation.  Thus $p_{ij}\leq q_i/i$.
Together with the initial condition that $q_i=1$, these constraints
completely characterize $K$.

{\bf Mapping.}  Recall a strategy $\pi$ indicates for each
$i$ and $j\leq i$ the {\it conditional} probability of making an offer
given that the employer is interviewing the $i$'th candidate and his
projected rank is $j$ whereas $p_{ij}$ is the {\it total} probability
of interviewing the $i$'th candidate with a projected rank of $j$ and
making an offer.  Thus $\pi(i,j)=p_{ij}/(q_i/i)$ and so
$p_{ij}=q_i\pi(i,j)/i$.  Together with the equailities derived above
that $q_1=1$ and $q_{i+1}=q_i-\sum_{j=1}^np_{ij}$, we can recursively
map any strategy $\pi$ to $K$ efficiently.  To map back we just take
the inverse of this bijection: given a point $\{p_{ij}\}$ in $K$, we
compute the (unique) $q_i$ satisfying the constraints $q_1=1$ and
$q_{i+1}=q_i-\sum_{j=1}^np_{ij}$, and define
$\pi(i,j)=p_{ij}/(q_i/i)$.

{\bf Payoff Matrix.}  By the above definitions, for any strategy $\pi$
and corresponding mapping $\{p_{ij}\}$, the probability that the
strategy hires the $j$'th best so far on round $i$ is $p_{ij}$.  Given
that employer $A$ hires the $j$'th best so far on round $i$ and
employer $B$ hires the $j'$'th best so far on round $i'$, we define
$M_{iji'j'}$ to be the probability that the overall rank of employer
$A$'s hire beats that of employer $B$'s hire plus one-half times the
probability that their ranks are equal.  We can derive the entries of
the this matrix as follows: Let $E^X_r$ be the event that with respect
to permutation $\sigma_X$ the overall rank of a fixed candidate is
$r$, and $F^X_{ij}$ be the event that the projected rank of the last
candidate in a random prefix of size $i$ is $j$.  Then
$$M_{iji'j'}=\sum_{r,r':1\leq r<r'\leq
  n}\Pr[E^A_r|F^A_{ij}]\Pr[E^B_{r'}|F^B_{i'j'}]+\frac{1}{2}\sum_{1\leq r\leq
  n}\Pr[E^A_r|F^A_{ij}]\Pr[E^B_r|F^B_{i'j'}].$$
Furthermore, by Bayes rule,
$\Pr[E^X_r|F^X_{ij}]=\Pr[F^X_{ij}|E^X_r]\Pr[E^X_r]/\Pr[F^X_{ij}]$
where $\Pr[E^X_r]=1/n$ and $\Pr[F^X_{ij}]=1/i$.  To compute
$\Pr[F^X_{ij}|E^X_r]$, we select the ranks of the other candidates in
the prefix of size $i$.  There are ${r-1\choose j-1}$ ways to pick the
ranks of the better candidates and ${n-r+1\choose i-j}$ ways to pick
the ranks of the worse candidates.  As there are ${n-1\choose i-1}$
ways overall to pick the ranks of the other candidates, we see:
$$\Pr[F^X_{ij}|E^X_r]=\frac{{r-1\choose j-1}{n-r+1\choose i-j}}{{n-1\choose i-1}}.$$
Letting $\{p_{ij}\}$ be the mapping $\phi(\pi)$ of employer $A$'s
strategy $\pi$ and $\{p'_{ij}\}$ be the mapping $\phi(\pi)$ of
employer $B$'s strategy $\pi'$, we see that
$c(\pi,\pi')=\phi(\pi)^tM\phi(\pi')$, as required.

By the above arguments, and the machinery from
Section~\ref{sec:bilinear}, we have proven
Theorem~\ref{thm:separatehiring} which claims that the equilibrium of
the independent hiring duel is computable.

\section{Compression Duel}
\label{sec:compression}

In a compression duel, two competitors each choose a binary tree with leaf set $\Omega$.  An element $\omega \in \Omega$ is then chosen according to distribution $p$, and whichever player's tree has $\omega$ closest to the root is the winner.  This game can be thought of as a competition between prefix-free compression schemes for a base set of words.
The Huffman algorithm, which repeatedly pairs nodes with lowest probability, is known to be optimal for single-player compression.

The compression duel is $D(X,\Omega,c,p)$, where $\Omega = [n]$ and $X$ is the set of binary trees with leaf set $\Omega$.
For $T \in X$ and $\omega \in \Omega$, $c(T,\omega)$ is the depth of $\omega$ in $T$.
In Section \ref{sec.compress.fail} we consider a variant in which not every element of $\Omega$ must appear in the tree.


\subsection{Computing an equilibrium}

The compression duel can be represented as a bilinear game.   In this case, $K$ and $K'$ will be sets of stochastic matrices, where a matrix entry $\{x_{ij}\}$ indicates the probability that item $\omega_i$ is placed at depth $j$.  The set $K$ is precisely the set of probability distributions over node depths that are consistent with probability distributions over binary trees.
We would like to compute minmax optimal algorithms as in Section \ref{sec:blexact}, but we do not have a randomized rounding scheme that maps elements of $K$ to binary trees.  Instead, following Section \ref{sec:appx}, we will find approximate minmax strategies by constructing an $\eps$-best response oracle.


The mapping $\phi : X \to K$ is straightforward: it maps a binary tree to its depth profile.  Also, the expected payoff of $x \in K$ versus $x' \in K'$ is
$\sum_i p(i) \sum_j x_{ij} \left( \frac{1}{2} x'_{ij} + \sum_{k > j} x'_{ij} \right)$
which can be written as $x^t M x'$ where matrix $M$ has bounded entries.  To apply Lemma \ref{lem:appx}, we must now provide an $\eps$ best response oracle, which we implement by reducing to a knapsack problem.

%
%
%
Fix $p$ and $x' \in K'$.  We will reduce the problem of finding a best response for $x'$ to the multiple-choice knapsack problem (MCKP), for which there is an FPTAS \cite{Lawler-79}.  In the MCKP, there are $n$ lists of items, say $\{ (\alpha_{i1}, \dotsc, \alpha_{ik_i})\ |\ 1 \leq i \leq n \}$, with each item $\alpha_{ij}$ having a value $v_{ij} \geq 0$ and weight $w_{ij} \geq 0$.  The problem is to choose exactly one item from each list with total weight at most $1$, with the goal of maximizing total value.  
Our reduction is as follows.  For each $\omega_i \in \Omega$ and $0 \leq j \leq n$, define $w_{ij} = 2^{-j}$ and
$v_{ij} = p(\omega_i) \left( \frac{1}{2} x'_{ij} + \sum_{d > j}x'_{id} \right)$.
This defines a MCKP input instance.  For any given $t \in X$,
$v(\phi(t),x') = \sum_{\omega_i \in \Omega} v_{i d_t(i)}$
and
$\sum_{\omega_i \in \Omega} w_{i, d_t(i)} \leq 1$
by the Kraft inequality.  Thus, any strategy for the compression duel can be mapped to a solution to the MCKP.  Likewise, a solution to the MCKP can be mapped in a value-preserving way to a binary tree $t$ with leaf set $\Omega$, again by the Kraft inequality.  This completes the reduction.
%

\subsection{Beatability}

We will obtain a bound of $3/4$ on the beatability of the Huffman algorithm.  The high-level idea is to choose an arbitrary tree $T$ and consider the leaves for which $T$ beats $H$ and vice-versa.  We then apply structural properties of trees to limit the relative sizes of these sets of leaves, then use properties of Huffman trees to bound the relative probability that a sampled leaf falls in one set or the other.  

Before bounding the beatability of the Huffman algorithm in the No Fail compression model, we review some facts about Huffman trees.  Namely, that nodes with lower probability occur deeper in the tree, and that siblings are always paired in order of probability (see, for example, page 402 of Gersting \cite{gersting-93}.  In what follows, we will suppose that $H$ is a Huffman tree.

\begin{fact}
\label{fact.huff.depths}
If $d_H(v_1) > d_H(v_2)$ then $p_H(v_1) \leq p_H(v_2)$.
\end{fact}

\begin{fact}
\label{fact.huff.sibling}
If $v_1$ and $v_2$ are siblings with $p_H(v_1) \leq p_H(v_2)$, then for every node $v_3 \in H$ either $p_H(v_3) \leq p_H(v_1)$ or $p_H(v_3) \geq p_H(v_2)$.
\end{fact}

We next give a bound on the relative probabilities of nodes on any given level of a Huffman tree, subject to the tree not being too ``sparse'' at the subsequent (deeper) level.  Let $p_H^{min}(d) = \min_{v : d_H(v) = d}p_H(v)$ and $p_H^{max}(d) = \max_{v : d_H(v) = d}p_H(v)$.
\begin{lemma}
\label{lem.huff.mult3}
Choose any $d < \max_v d_H(v)$ and nodes $v, w$ such that $d_H(w) = d_H(v) = d$.  If $v$ is not the common ancestor of all nodes of depth greater than $d$, then $p_H(w) \leq 3p_H(v)$.
\end{lemma}

\begin{proof}
Let $a = p_H(v)$.  By assumption there exists a non-leaf node $z \neq v$ with $d_H(z)=d$, say with children $z_1$ and $z_2$.
Then $p_H(z_1) \leq a$ and $p_H(z_2) \leq a$ by Fact \ref{fact.huff.depths}, so $p_H(z) \leq 2a$.  This implies that $v$'s sibling has probability at most $2a$ by Fact \ref{fact.huff.sibling}, so the parent of $v$ has probability at most $3a$.  Fact \ref{fact.huff.depths} then implies that $p_H(w) \leq 3a$ as required.
\end{proof}

For any $T \in X$ and set of nodes $R \subseteq T$ we define the weight of $R$ to be $w_T(R) = \sum_{v \in R}2^{-d_T(v)}$.  The Kraft inequality for binary trees is $w_T(T) \leq 1$.  In fact, we have $w_T(T) = 1$ since we can assume each interior node of $T$ has two children.

\begin{lemma}
\label{lem.huff.weight}
Choose $R \subseteq H$ such that no node of $R$ is a descendent of any other, and suppose $w(R) = 2^{-d}$ for some $d \in [n]$.  Then $p^{min}_H(d) \leq p(R) \leq p^{max}_H(d)$.
\end{lemma}
\begin{proof}
We will show $p(R) \leq p^{max}_H(d)$; the argument for the other inequality is similar.  We proceed by induction on $|R|$.  If $|R| = 1$ the result is trivial (since $R = \{v\}$ where $d_H(v) = d$).  Otherwise, since $w(R) = 2^{-d}$, there must be at least two nodes of the maximum depth present in $R$.  Let $v$ and $w$ be the two such nodes with smallest probability, say with $p_H(v) \leq p_H(w)$.  Let $w'$ be the parent of $w$.  Then $p_H(w') \geq p_H(w) + p_H(v)$, since the sibling of $w$ has weight at least $p_H(v)$ by Fact \ref{fact.huff.sibling}.  Also, $w' \not\in R$ since $w \in R$ and no node of $R$ is a descendent of any other.  Let $R' = R \cup \{w'\} - \{w,v\}$.  Then $w(R') = w(R)$, $p(R') \geq p(R)$, and no node of $R'$ is a descendent of any other.  Thus, by induction, $p(R) \leq p(R') \leq p^{max}_H(d)$ as required.
\end{proof}

We are now ready to show that the beatability of the Huffman algorithm is at most $\frac{3}{4}$.

\begin{proposition}
\label{prop.huffman.beatable}
The beatability of the Huffman algorithm is at most $\frac{3}{4}$.
\end{proposition}

Fix $\Omega$ and $p$.  Let $H$ denote the Huffman tree and choose any other tree $T$.  Define $P = \{ v \in \Omega : d_{T}(v) < d_{H}(v) \}$, $Q = \{ v \in \Omega : d_{T}(v) > d_{H}(v) \}$.  That is, $P$ is the set of elements of $\Omega$ for which $T$ beats $H$, and $Q$ is the set of elements for which $H$ beats $T$.  Our goal is to show that $p(P) < 3p(Q)$, which would imply that $v(T,H) \leq 3/4$.

We first claim that $w(P) < w(Q)$.  To see this, write $U = \Omega -( P\cup Q)$ and note that, by the Kraft inequality,
\begin{equation}
\label{eq.kraft.1}
w(P) + w(Q) + w(U) = 1 = w_{T}(P) + w_{T}(Q) + w_{T}(U).
\end{equation}
Moreover, $w_{T}(Q) > 0$, $w_{T}(U) = w_H(U)$, and $w_{T}(P) \geq 2w(P)$ (since $d_{T}(v) \leq d_H(v) - 1$ for all $v \in P$).   Applying these inequalities to \eqref{eq.kraft.1} implies $w(P) - w(Q) < 0$, completing the claim.

Our approach will be to express $P$ and $Q$ as disjoint unions $P = P_1 \cup \dotsc \cup P_r$ and $Q = Q_1 \cup \dotsc \cup Q_r$ such that $p(P_i) \leq 3p(Q_i)$ for all $i$.  To this end, we express the quantities $w(P)$ and $w(Q)$ in binary: choose $x_1, \dotsc, x_n$ and $y_1, \dotsc, y_n$ from $\{0,1\}$ such that $w(P) = \sum_i x_i 2^{-i}$ and $w(Q) = \sum_i y_i 2^{-i}$.  Since $w(P)$ is a sum of element weights that are inverse powers of two, we can partition the elements of $P$ into disjoint subsets $P_1, \dotsc, P_n$ such that $w(P_i) = x_i 2^{-i}$ for all $i \in [n]$.  Similarly, we can partition $Q$ into disjoint subsets $Q_1, \dotsc, Q_n$ such that $w(Q_i) = y_i 2^{-i}$ for all $i \in [n]$.

Let $r = \min\{ i : x_i \neq y_i \}$.  Note that, since $w(P) < w(Q)$, we must have $x_r = 0$ and $y_r = 1$.

We first show that $p(P_i) \leq 3p(Q_i)$ for each $i < r$.  Since $x_i = y_i$, we either have $P_i = Q_i = \emptyset$ or else $w(P_i) = w(Q_i) = 2^{-i}$.  In the latter case, suppose first that $|Q_i|=1$.  Then, since $Q_i$ consists of a single leaf and $i$ is not the maximum depth of tree $H$, we can apply Lemma \ref{lem.huff.weight} and Lemma \ref{lem.huff.mult3} to conclude
$p(P_i) \leq p^{max}_H(i) \leq 3p(Q_i)$.
Next suppose that $|Q_i| > 1$.  We would again like to apply Lemma \ref{lem.huff.mult3}, but we must first verify that its conditions are met.  Suppose for contradiction that all nodes of depth greater than $i$ share a common ancestor of depth $i$. Then, since $w(Q_i) = 2^{-i}$ and $|Q_i| > 1$, it must be that $Q_i$ contains all such nodes, which contradicts the fact that $Q_r$ contains at least one node of depth greater than $i$. We conclude that the conditions of Lemma \ref{lem.huff.mult3} are satisfied for all $v$ and $w$ at depth $i$, and therefore
$p(P_i) \leq p^{max}_H(i) \leq 3p^{min}_H(i) \leq 3p(Q_i)$
as required.

We next consider $i \geq r$.  Let $P'_r = \bigcup_{j \geq r} P_j$ and $Q'_r = \bigcup_{j \geq r} Q_j$.  We claim that $p(P'_r) \leq 3p(Q'_r)$.  If $P'_r = \emptyset$ then this is certainly true, so suppose otherwise.  Then $w(P'_r) < 2^{-r}$, so $P'_r$ contains elements of depth greater than $r$.  As in the case $i < r$, this implies that either $Q_r$ contains only a single node (and cannot be the common ancestor of all nodes of depth greater than $r$), or else not all nodes of depth greater than $r$ have a common ancestor of depth $r$.  We can therefore apply Lemma \ref{lem.huff.weight} and Lemma \ref{lem.huff.mult3} to conclude
$p(P'_r) \leq p^{max}_H(r) \leq 3p(Q_r) \leq 3p(Q'_r)$.

Since $P = P_1 \cup \dotsc \cup P_{r-1} \cup P_r'$ and $Q = Q_1 \cup \dotsc \cup Q_{r-1} \cup Q_r'$ are disjoint partitions, we conclude that $p(P) \leq 3p(Q)$ as required.
\qed

We now give an example to demonstrate that the Huffman algorithm is at least $(2/3 - \epsilon)$-beatable for every $\epsilon > 0$.  For any $n \geq 3$, consider the probability distribution given by $p(\omega_1) = \frac{1}{3}$, $p(\omega_i) = \frac{1}{3 \cdot 2^{i-2}}$ for all $1 < i < n$, and $p(\omega_n) = \frac{1}{3 \cdot 2^{n-3}}$.  For this distribution, the Huffman tree $t$ satisfies $d_t(\omega_i) = i$ for each $i < n$ and $d_t(\omega_n) = n-1$.  Consider the alternative tree $t'$ in which $d(\omega_1) = n-1$ and $d(\omega_i) = i-1$ for all $i > 1$.  Then $t'$ will win if any of $\omega_2, \omega_3, \dotsc, \omega_{n-1}$ are chosen, and will tie on $\omega_n$.  Thus
$v(t',t) = \sum_{i > 1}\frac{1}{3 \cdot 2^{i-2}} + \frac{1}{2} \cdot \frac{1}{3\cdot 2^{n-3}} = \frac{2}{3} - \frac{1}{3 \cdot 2^{n-2}}$,
and hence the Huffman algorithm is $(\frac{2}{3} - \frac{1}{3 \cdot 2^{n-2}})$-beatable for every $n \geq 3$.

We conclude the section by noting that if all probabilities are inverse powers of $2$, the Huffman algorithm is minmax optimal. \

\begin{proposition}
\label{prop.huffman.power2}
Suppose there exist integers $a_1, \dotsc, a_n$ such that $p(\omega_i) = 2^{-a_i}$ for each $i \leq n$.  Then the value of the Huffman tree $H$ is $v(H) = 1/2$.
\end{proposition}

\begin{proof}
We suppose that there exist integers $a_1, \dotsc, a_n$ such that $p(\omega_i) = 2^{-a_i}$ for each $i \leq n$.  Our goal is to show that the value of the Huffman tree $H$ is $v(H) = 1/2$.

For this set of probabilities, the Huffman tree will set $d_H(\omega_i) = a_i$ for all $\omega_i \in \Omega$.  In this case, $p(R) = w(R)$ for all $R \subseteq H$.   Choose any other tree $T$, and define sets $P$ and $Q$ as in the proof of Proposition \ref{prop.huffman.beatable}.  That is, $P$ is the set of elements of $\Omega$ for which $T$ beats $H$, and $Q$ is the set of elements for which $H$ beats $T$.  Then, as in Proposition \ref{prop.huffman.beatable}, we must have $w(P) < w(Q)$, and hence $p(P) < p(Q)$.  Thus $v(H,T) < 1/2$.  We conclude that the best response to the Huffman tree $H$ must be $H$ itself, and thus strategy $H$ has a value of $1/2$.
\end{proof}


\subsection{Variant: allowed failures}
\label{sec.compress.fail}

We consider a variant of the compression duel in which an algorithm
can fail to encode certain elements.  If we write $L(T)$ to be the set
of leaves of binary tree $T$, then in the (original) model of
compression we require that $L(T) = \Omega$ for all $T \in X$, whereas
in the ``Fail" model we require only that $L(T) \subseteq \Omega$.  If
$\omega \not\in L(T)$, we will take $c(T,\omega) = \infty$.  The
Huffman algorithm is optimal for single-player compression in the Fail
model.

We note that our method of computing approximate minmax algorithms
carries over to this variant; we need only change our best-response
reduction to use a Multiple-Choice Knapsack Problem in which \emph{at
  most} one element is chosen from each list.  What is different,
however, is that the Huffman algorithm is completely beatable in the
Fail model.  If we take $\Omega = \{\omega_1,\omega_2\}$ with
$p(\omega_1) = 1$ and $p(\omega_2) = 0$, the Huffman tree $H$ places
each of the elements of $\Omega$ at depth $2$.  If $T$ is the
singleton tree that consists of $\omega_1$ as the root, then $v(T, H)
= 1$.

\section{Binary Search Duel}
\label{sec:bst}

In a binary search duel, $\Omega = [n]$ and $X$ is the set of binary
search trees on $\Omega$ (i.e. binary trees in which nodes are
labeled with elements of $\Omega$ in such a way that an in-order
traversal visits the elements of $\Omega$ in sorted order). Let $p$ be
a distribution on $\Omega$. Then for $T \in X$ and $\omega \in
\Omega$, $c(T, \omega)$ is the depth of the node labeled by
``$\omega$'' in the tree $T$. In single-player binary search and
uniform $p$,
selecting the median
$m$ element in $\Omega$ as the root node and recursing on
the left $\{\omega | \omega < m\}$ and right $\{\omega | \omega > m\}$
subsets to construct
sub-trees is known to be
optimal.

The binary search game can be represented as a bilinear duel. In this
case, $K$ and $K'$ will be sets of stochastic matrices (as in the case
of the compression game) and the entry $\{ x_{i, j}\}$ will represent
the probability that item $\omega_j$ is placed at depth $i$. Of
course, not every stochastic matrix is realizable as a distribution
on binary search trees (i.e.\
such that the probability $\omega_j$ is placed at depth $i$ is $\{x_{i, j}\}$).
In order to define linear constraints on $K$
so that any matrix in $K$ is realizable, we will introduce an
auxiliary data structure in Section~\ref{sec:sas} called the {\sc
  State-Action Structure}
that
captures the decisions made by a binary search tree. Using these
ideas, we will be able to fit the binary search game into the bilinear
duel framework introduced in Section~\ref{sec:blexact} and hence be
able to efficiently compute a Nash equilibrium strategy for each
player.

Given a binary search tree $T \in X$, we will write $c_T(\omega)$ for
the depth of $\omega$ in $T$. We will also refer to $c_T(\omega)$ as
the time that $T$ finds $\omega$.

\subsection{Computing an equilibrium}\label{sec:sas}

In this subsection, we give an algorithm for computing a Nash
equilibrium for the binary search game,
based on the
bilinear duel framework introduced in Section~\ref{sec:blexact}. We
will do this by defining a structure called the {\sc State-Action
  Structure} that we can use to represent the decisions made by a
binary search tree using only polynomially many variables. The set of
valid variable assignments in a {\sc State-Action Structure} will also
be defined by only polynomially many linear constraints and so these
structures will naturally be closed under taking convex
combinations.
We will demonstrate that the value of
playing $\sigma \in \Delta(X)$ against any value matrix $V$ -- see
Definition~\ref{def:penalty} is a linear function of the variables in
the {\sc State-Action Structure} corresponding to $\sigma$.
Furthermore, all valid {\sc State-Action Structures} can be
efficiently realized as a distribution on binary search trees which
achieves the same expected value.

To apply the bilinear duel framework, we must give a mapping $\phi$
from the space of binary search trees to a convex set $K$ defined
explicitly by a polynomial number of linear constraints (on a
polynomial number of variables). We now give an informal description
of $K$: The idea is to represent a binary search tree $T \in X$ as a
layered graph. The nodes (at each depth) alternate in type. One layer
represents the current knowledge state of the binary search
tree. After making some number of queries (and not yet finding the
token), all the information that the binary search tree knows is an
interval of values to which the token is confined - we refer to this
as the \emph{live interval}. The next layer of nodes represents an
action - i.e. a query to some item in the live
interval. Correspondingly, there will be three outgoing edges from an
action node representing the possible replies that either the item is
to the left, to the right, or at the query location (in which case the
outgoing edge will exit to a terminal state).

We will define a flow on this layered graph based on $T$ and the
distribution $p$ on $\Omega$. Flow will represent total probability -
i.e. the total flow into a state node will represent the probability
(under a random choice of $\omega \in \Omega$ according to $p$) that
$T$ reaches this state of knowledge (in exactly the corresponding
number of queries). Then the flow out of a state node represents a
decision of which item to query next. And lastly, the flow out of an
action node splits according to Bayes' Rule - if all the information
revealed so far is that the token is confined to some interval, we can
express the probability that (say) our next query to a particular item
finds the token as a conditional probability. We can then take convex
combinations of these "basic" flows
in order to form flows corresponding to
distributions on
binary search trees.

We give a randomized rounding algorithm to select a random binary
search tree based on a flow - in such a way that the marginal
probabilities of finding a token $\omega_i$ at time $r$ are exactly
what the flow specifies they should be. The idea is that if we choose
an outgoing edge for each state node (with probability proportional to
the flow), then we have fixed a binary search tree because we have
specified a decision rule for each possible internal state of
knowledge. Suppose we were to now select an edge out of each action
node (again with probability proportional to the flow) and we were to
follow the unique path from the start node to a terminal node. This
procedure would be equivalent to searching for a randomly chosen token
$\omega_i$ chosen according to $p$ and using this token to choose
outgoing edges from action nodes. This procedure generates a random
path from the start node to a terminal node, and is in fact equivalent
to sampling a random path in the path decomposition of the flow
proportionally to the flow along the path. Because these two rounding
procedures are equivalent, the marginal distribution that results from
generating a binary search tree (and choosing a random element to look
for) will exactly match the corresponding values of the flow.

\subsection{Notation}

The natural description of the strategy space of the binary search game is exponential (in $|\Omega|$) -- so we will assume that the value of playing any binary search tree $T$ against an opponent's mixed strategy is given to us in a compact form which we will refer to as a value matrix:

\begin{definition}~\label{def:penalty}
A value matrix $V$ is an $|\Omega| \times |\Omega|$ matrix in which the entry $V_{i, j}$ is interpreted to be the value of finding item $\omega_j$ at time $i$.
\end{definition}

Given any binary search tree $T' \in X$, we can define a value matrix $V(T')$ so that the expected value of playing any binary search tree $T \in X$ against $T$ in the binary search game can be written as $\sum_{i, j} 1_{c_T(\omega_j) = i} V(T')_{i, j}$:

\begin{definition}
Given a binary search tree $T' \in X$, let $V(T')$ be a value matrix such that
\[ V(T')_{i, j} = \left\{ \begin{array}{ll}
         0 & \mbox{if } c_{T'}(\omega_j) < i \\
         \frac{1}{2} & \mbox{if } c_{T'}(\omega_j) = i \\
        1 &  \mbox{if } c_{T'}(\omega_j) > i \end{array} \right. \]
Similarly, given a mixed strategy $\sigma' \in \Delta(X)$, let $V(\sigma') = E_{T' \sim \sigma'}[V(T')]$
\end{definition}

Note that not every value matrix $V$ can be realized as the value matrix $V(T')$ for some $T' \in X$. In fact, $V$ need not be realizable as $V(\sigma)$ for some $\sigma \in \Delta(X)$. However, we will be able to compute the best response against any value matrix $V$, regardless of whether or not the matrix corresponds to playing the binary search game against an adversary playing some mixed strategy. Lastly, we define a stochastic matrix $I(T)$, given $T \in X$. From $I(T)$, and $V(T')$ we can write the expected value of playing $T$ against $T'$ as a inner-product. We let $< A, B>_{p} = \sum_{i, j} A_{i, j} B_{i, j} p(\omega_j)$ when $A$ and $B$ are $|\Omega| \times |\Omega|$ matrices.

\begin{definition}
Given a binary search tree $T \in X$, let $I(T)$ be an $|\Omega| \times |\Omega|$ matrix in which $I(T)_{i, j} = 1_{c_T(\omega_j) = i}$. Similarly, given $\sigma \in \Delta(X)$, let $I(\sigma) = E_{T \sim \sigma}[I(T)]$.
\end{definition}

\begin{lemma}
Given $\sigma, \sigma' \in \Delta(X)$, the expected value of playing $\sigma$ against $\sigma'$ in the binary search game is exactly $<I(\sigma), V(\sigma')>_p$.
\end{lemma}

\begin{proof}
Consider any $T, T' \in X$. Then the expected value of playing $T$ against $T$ in the binary search game is exactly $\sum_{i} p(\omega_i) \Big [1_{c_{T}(\omega_i) < c_{T'}(\omega_i)} + \frac{1}{2} 1_{c_{T}(\omega_i) = c_{T'}(\omega_i)} \Big ] = <I(T), V(T')>_p$. And since $<I(T), V(T')>_p$ is bilinear in the matrices $I(T)$ and $V(T')$, indeed the expected value of playing $\sigma$ against $\sigma'$ is $<I(\sigma), V(\sigma')>_p$.
\end{proof}

\subsection{{\sc State-Action Structure}}

\begin{definition}~\label{def:lsr}
Given a distribution $p$ on $\Omega$ and $\omega_i, \omega_j, \omega_k
\in \Omega$ (and $ai < j < k$), let
$$p_{i, j, k}^L = \frac{ Pr_{\omega_{k'} \sim p}[ i \leq k' < k ]}{ Pr_{\omega_{k'} \sim p}[ i \leq k' \leq j ]}, p_{i, j, k}^E = \frac{ Pr_{\omega_{k'} \sim p}[ k' = k ]}{ Pr_{\omega_{k'} \sim p}[ i \leq k' \leq j ]}, \mbox{ and } p_{i, j, k}^R = \frac{ Pr_{\omega_{k'} \sim p}[ k < k' \leq j ]}{ Pr_{\omega_{k'} \sim p}[ i \leq k' \leq j ]}$$
\end{definition}

Intuitively, we can regard the interval $[\omega_i, \omega_j]$ as
being divided into the sub-intervals $[\omega_i, \omega_{k-1}]$,
$\{\omega_k\}$ and $[\omega_{k+1}, \omega_j]$. Then the quantity
$p_{i, j, k}^L$ represents the probability that randomly generated
element is contained in the first interval, conditioned on the element
being contained in the original interval $[\omega_i,
  \omega_j]$. Similarly, one can interpret $p_{i, j, k}^E$ and $p_{i,
  j, k}^R$ as being conditional probabilities as well.

We also define a set of knowledge states, which represent the current
information that the binary search tree knows about the element and
also how many queries have been made:

\begin{definition}
We define:
\begin{enumerate}
\item $\cS = \{ (i, j, r) | \omega_i, \omega_j \in \Omega, i < j,
  \mbox{ and } r \in \{1, 2, ...., |\Omega|\}\}$
\item $\cA = \{ (S, k) | S = (i, j, r) \in \cS, \omega_k \in \Omega
  \mbox{ and } k \in (i, j)\}$
\item $\cF = \{ (k, r) | \omega_k \in \Omega \mbox{ and } r \in \{1,
  2, ...., |\Omega|\}\}$
\end{enumerate}
We will refer to $\cS$ as the set of knowledge state. Additionally we
will refer to $S_{start} = (\omega_1, \omega_n, 0)$ as the start
state. We will refer to $\cA$ as the set of action state and $\cF$ as
the set of termination states.
\end{definition}

We can now define a {\sc State-Action Structure}:

\begin{definition}\label{def:sastructure}
A {\sc State-Action Structure} is a fixed directed graph generated as:
\begin{enumerate}
\item Create a node $n_S$ for each $S \in \cS$, a node $n_A$ for each
  $A \in \cA$ and a node $n_F$ for each $F \in \cF$.
\item For each $S = (i, j, r) \in \cS$, and for each $k$ such that $i
  < k < j$, create a directed edge $e_{S, k}$ from $S$ to $A = (S, k)
  \in \cA$.
\item For each $A = (S, k) \in \cA$ and $S = (i, j, r)$, create a
  directed edge $e_{A, F}$ from $A$ to $F = (k, r+ 1)$ and directed
  edges $e_{A, S_L}$ and $e_{A, S_R}$ from $A$ to $S_L$ and $S_R$
  respectively for $S_L = (i, k-1, r+1)$ and $S_R = (k + 1, j, r+ 1)$.
\end{enumerate}
\end{definition}

We will define a flow on this directed graph. The source of this flow
will be the start node $S_{start}$ and the node corresponding to each
termination state will be a sink. The total flow in this graph will be
one unit, and this flow should be interpreted as representing the
total probability of reaching a particular knowledge state, or
performing a certain action.

\begin{definition}\label{def:stateful}
We will call an set of values $x_{e}$ for each directed edge in a {\sc
  State-Action Structure} a stateful flow if (let us adopt the
notation that $x_{S, A}$ is the flow on an edge $e_{S, A}$):
\begin{enumerate}
\item For all $e$, $0 \leq x_e \leq 1$
\item All nodes except $n_{S_{start}}$ and $n_{F}$ (for $F \in \cF$)
  satisfy conservation of flow
\item For each action state $A = (S, i) \in \cA$ for $S = (i, j, r)$,
  the the flow on the three out-going edges $e_{A, F}, e_{A, S_L}$ and
  $e_{A, S_R}$ from $n_A$, satisfy $x_{A, F} = p_{i, j, k}^E C$,
  $x_{A, S_L} = p_{i, j, k}^L C$ and $x_{A, S_R} = p_{i, j, k}^R$
  where $C = \sum_{e = (S', A) \mbox{ for } S' \in \cS} x_{S', A}$
\end{enumerate}
\end{definition}

Given $T \in X$, we can define a flow $x_T$ in the {\sc State-Action
  Structure} that captures the decisions made by $T$:
\begin{definition}\label{def:xt}
Given $T \in X$, define $x_T$ as follows:
\begin{enumerate}
\item For each $S = (i, j, r) \in \cS$ let $T_{i, j}$ be the sub-tree
  of $T$ (if a unique such sub-tree exists) such that the labels
  contained in $T_{i, j}$ are exactly $\{\omega_i, \omega_{i+1}, ...,
  \omega_{j}\}$. Suppose that the root of this sub-tree $T_{i, j}$ is
  $\omega_k$. Then send all flow entering the node $n_S$ on the
  outgoing edge $e_{S, A}$ for $A = (S, k)$.
\item For each $A \in \cA$, divide flow into a action node $n_A$
  according to Condition $3$ in Definition~\ref{def:stateful} among
  outgoing edges.
\end{enumerate}
\end{definition}

Note that the flow out of $n_{S_{start}}$ is one. Of course, the
choice of how to split flow on outgoing edges from an action node
$n_A$ is already well-defined. But we need to demonstrate that $x_T$
does indeed satisfy conservation of flow requirements, and hence is a
stateful flow:

\begin{lemma}~\label{lemma:isstateful}
For any $T \in X$, $x_T$ is a stateful flow
\end{lemma}

\begin{proof}
For some intervals $\{\omega_i, \omega_{i+1}, ..., \omega_{j}\}$,
there is no sub-tree in $T$ for which the labels contained in the
sub-tree is exactly $\{\omega_i, \omega_{i+1}, ..., \omega_{j}\}$. If
there is such an interval, however, it is clearly unique. We will
prove by induction that the only state nodes in the {\sc State-Action
  Structure} which are reached by flow $x_T$ are state nodes for which
there is such a sub-tree.

We will prove this condition by induction on $r$ for state nodes $n_S$
of the form $S = (i, j, r)$. This condition is true in the base case
because all flow starts at the node $n_{S_{start}}$ and $S_{start} =
(\omega_1, \omega_n, 0)$ and indeed the entire binary search tree $T$
has the property that the set of labels used is exactly $\{\omega_1,
\omega_2, ... \omega_n\}$.

Suppose by induction that there is some sub-tree $T_{i, j}$ of $T$ for
which the labels of contained in the sub-tree are exactly $\{\omega_i,
\omega_{i+1}, ..., \omega_{j}\}$. Let $\omega_k$ be the label of the
root node of $T_{i, j}$. Then all flow entering $n_S$ would be sent to
the action node $A = (S, k)$ and all flow out of this action node
would be set to either a termination node or to state nodes $S_L = (i,
k-1, r+1)$ or $S_R = (k+1, r+1)$ and both of the intervals
$\{\omega_i, \omega_2, ... \omega_{r-1}\}$ or $\{\omega_{r+1},
\omega_{r+2}, ..., \omega_{j}\}$ do indeed have the property that
there is a sub-tree that contains exactly each respective set of
labels - these are just the left and right sub-trees of $T_{i, j}$.
\end{proof}

The variables in a stateful flow capture marginal probabilities that
we need to compute the expected value of playing a binary search tree
$T$ against some value matrix $V$:

\begin{lemma}~\label{lemma:conditional}
Consider any state $S = (i, j, r) \in \cS$. The total flow in $x_T$
into $n_S$ is exactly the probability that (under a random choice of
$\omega_k \sim p$), $\omega_k$ is contained in some sub-tree of $T$ at
depth $r+1$. Similarly the total flow in $x_T$ into any terminal node
$n_F$ for $F = (\omega_f, r)$ is exactly the probability (under a
random choice of $\omega_k \sim p$) that $c_T(\omega_k)=r$.
\end{lemma}

\begin{proof}
We can again prove this lemma by induction on $r$ for state nodes
$n_S$ of the form $S = (i, j, r)$.  In the base case, the flow into
$n_{S_{start}}$ is $1$, which is exactly the probability that (under a
random choice of $\omega_t \sim p$), $\omega_t$ is contained in some
sub-tree of $T$ at depth $1$.

So we can prove the inductive hypothesis by sub-conditioning on the
event that the element $\omega_k$ is contained in some sub-tree of $T$
at depth $r$. Let this subtree be $T'$. By the inductive hypothesis,
this is exactly the flow into the node $n_{S'}$ where $S' = (i, j,
r-1)$ for some $\omega_i, \omega_j \in \Omega$ and $i \leq k \leq
j$. We can then condition on the event that $\omega_k$ is such that $i
\leq k \leq j$. Let $\omega_r$ be the label of the root node of
$T'$. Then using conditioning, the probability that $\omega_k$ is
contained in the left-subtree of $T'$ is exactly $p_{i, j, r}^L$, and
similarly for the right sub-tree. Also the probability that $\omega_k
= \omega_r$ is $p_{i, j, r}^E$. And so Condition $3$ in
Definition~\ref{def:stateful} enforces the condition that the flow
splits exactly as this total probability splits - i.e. the probability
that $\omega_k$ is contained in the left and right sub-interval of
$\{\omega_i, \omega_{i+1}, ... \omega_{j}\}$ or contained in the root
"$\omega_r$" respectively. Note that the set of sub-trees at any
particular depth in $T$ correspond to disjoint intervals of $\Omega$,
and hence there is no other flow entering the state $n_S$, and this
proves the inductive hypothesis.
\end{proof}

As an immediate corollary:
\begin{corollary}\label{cor:vmat}
The expected value of playing $T$ against value matrix $V$, $$<I(T),
V>_p = \sum_{F = (\omega_k, r) \in \cF} x_T^{in}(F) V_{r, k}$$ where
$x_T^{in}$ denotes the total flow into a node according to $x_T$.
\end{corollary}

And as a second corollary:
\begin{corollary}\label{cor:mmat}
Given $T \in X$, $$V(T)_{i, j} = \frac{\frac{1}{2} x_T^{in}(\omega_j,
  i) + \sum_{i' > i} x_T^{in}( \omega_j, i')}{p(\omega_j)}$$ where
$x_T^{in}(\omega_j, i)$ denotes the total flow into $n_F$ for $F =
(\omega_j, i) \in \cF$.
\end{corollary}

\subsection{A rounding algorithm}

\begin{proposition}~\label{prop:rr}
Given a stateful flow $x$, there is an efficient randomized rounding
procedure that generates a random $T \in X$ with the property that for
any $\omega_j \in \Omega$ and for any $i \in \{1, 2, ..., |\Omega|\}$,
$Pr[c_T(\omega_j) = i] = \frac{x^{in}(\omega_j, i)}{p_{\omega_j}}$.
\end{proposition}

\begin{proof}
Since $x$ is a unit flow from $n_{S_{start}}$ to the set of sink nodes
$n_F$ for $F \in \cF$. So if we could sample a random path
proportional to the total flow along the path, the probability that
the path ends at any sink $n_F$ for $F = (\omega_j, r)$ is exactly
$x^{in}(\omega_j, r)$.

\textbf{First Rounding Procedure:} Consider the following procedure
for generating a path according to this distribution - i.e. the
probability of generating any path is exactly the flow along the path:
Starting at the source node, and at every step choose a new edge to
traverse proportionally to the flow along it. So if the process is
currently at some node $n_S$ and the total flow into the node is $U$,
and the total flow on some outgoing edge $e$ is $u$, edge $e$ is
chosen with probability exactly $\frac{u}{U}$ and the process
continues until a sink node is reached. Notice that this procedure
always terminates in $O(|\Omega|)$ steps because each time we traverse
an action node $n_A$, the counter $r$ is incremented and every edge in
a {\sc State-Action Structure} either points into or points out of a
action node.

The key to our randomized rounding procedure is an alternative way to
generate a path from the source node to a sink such that the
probability that the path ends at any sink $n_F$ for $F = (\omega_j,
r)$ is \emph{still} exactly $x^{in}(\omega_j, r)$. Instead, for each
state node $n_{S}$, we choose an outgoing edge in advance (to some
action node) proportional to the flow on $x$ on that edge.

\textbf{Second Rounding Procedure:} If we fix these choices in
advance, we can define an alternate path selection procedure which
starts at the source node, and traverse any edges that have already
been decided upon. Whenever the process reaches an action node (in
which case the outgoing edge has not been decided upon), we can select
an edge proportional to the total flow on the edge. This procedure
still satisfies the property that the probability that the path ends
at any sink $n_F$ for $F = (\omega_j, r)$ is exactly $x^{in}(\omega_j,
r)$.

\textbf{Third Rounding Procedure:} Next, consider another modification
to this procedure. Imagine still that the outgoing edges from every
state node are chosen (randomly, as above in the \textbf{Second
  Rounding Procedure:} ). Instead of choosing which outgoing edge to
pick from an action node when we reach it, we could instead pick an
item $\omega_{k'} \sim p$ in advance and using this hidden value to
determine which outgoing edge from a action node to traverse. We will
maintain the invariant that if we are at $n_A$ and $A = (S, k)$ for $S
= (i, j, r)$, we must have $i \leq k' \leq j$. This is clearly true at
the base case. Then we will traverse the edge $e_{A, F}$ for $F = (k,
r)$ if $\omega_{k'} = \omega_k$. Otherwise if $i \leq k' \leq k-1$ we
will traverse the edge $e_{A, S_L}$ for $S_L = (i, k-1,
r+1)$. Otherwise $i \leq k' \leq k-1$ and we will traverse the edge
$e_{A, S_R}$ for $S_R = (k+1, j, r+1)$. This clearly maintains the
invariant that $k'$ is contained in the interval corresponding to the
current knowledge state.

This third procedure is equivalent to the second procedure. This
follows from interpreting Condition $3$ in
Definition~\ref{def:stateful} as a rule for splitting flow that is
consistent with the conditional probability that $\omega_{k'}$ is
contained in the left or right sub-interval of $\{\omega_i,
\omega_{i+1}, ... \omega_{j}\}$ or is equal to $\omega_k$ conditioned
on $\omega_{k'} \in \{\omega_i, \omega_{i+1}, ... \omega_{j}\}$. An
identical argument is used in the proof of
Lemma~\ref{lemma:conditional}. In this case, we will say that
$\omega_{k'}$ is the rule for choosing edges out of action nodes.

Now we can prove the Lemma: The key insight is that once we have
chosen the outgoing edges from each state node (but not which outgoing
edges from each action node), we have determined a binary search tree:
Given any element $\omega_{k'}$, if we follow outgoing edges from
action nodes using $\omega_{k'}$ as the rule, we must reach a terminal
node $F = (\omega_{k'}, r)$ for some $r$. In fact, the value of $r$ is
determined by $\omega_{k'}$ because once $\omega_{k'}$ is chosen,
there are no more random choices. So we can compute a vector of
dimension $|\Omega|$, $\vec{u}$ such that $\vec{u}_j = r$ such that
$F= (\omega_{j}, r)$ is reached when the $\omega_j$ is the rule for
choosing edges out of action nodes.

Using the characterization in Proposition~\ref{prop:depth}, it is easy
to verify that the transition rules in the {\sc State Action
  Structure} enforce that $\vec{u}$ is a depth vector and hence we can
compute a binary search tree $T$ which has the property that using
selection rule $\omega_j$ results in reaching the sink node $F =
(\omega_j, c_T(\omega_j))$.

Suppose we select each outgoing edge from a state node (as in the
\textbf{Third Rounding Procedure}) and select an $\omega_{k'} \sim p$
(again as in the \textbf{Third Rounding Procedure})
independently. Then from the choices of the outgoing edges from each
state node, we can recover a binary search tree $T$. Then $Pr_{T,
  \omega_{k'}}[c_T(\omega_{k'}) = r] = x^{in}(\omega_{k'}, r)$
precisely because the \textbf{First Rounding Procedure} and the
\textbf{Third Rounding Procedure} are equivalent. And then we can
apply Bayes' Rule to compute that $$Pr_{T}[c_T(\omega_{k'}) = r |
  \omega_{k'} = \omega_k] = \frac{x^{in}(\omega_{k},
  r)}{p(\omega_{k})}$$
\end{proof}

\begin{theorem}
There is an algorithm that runs in time polynomial in $|\Omega|$ that
computes an exact Nash equilibrium for the binary search game.
\end{theorem}

\begin{proof}
We can now apply the biliear duel framework introduced in
Section~\ref{sec:blexact} to the binary search game: The space $K$ is
the set of all stateful flows. The set of variables is polynomially
sized -- see Definition~\ref{def:sastructure}, and the set of linear
constraints is also polynomially sized and is given explicitly in
Definition~\ref{def:stateful}. The function $\phi$ maps binary search
trees $T \in X$ to a stateful flow $x_T$ and is the procedure given in
Defintion~\ref{def:xt} for computing this mapping is efficient. Also
the payoff matrix $M$ is given explicitly in Corollary~\ref{cor:vmat}
and Corollary~\ref{cor:mmat}. And lastly we give a randomized rounding
algorithm in Proposition~\ref{prop:rr}.
\end{proof}

\subsection{Beatability}

We next consider the beatability of the classical algorithm when $p$ is the uniform distribution on $\Omega$. For lack of a better term, let us call this
single-player optima the median binary search - or median search.

Here we give matching upper and lower bounds on the beatability of median search. The idea is that an adversary attempting to do well against median search can only place one item at depth $1$, two items at depth $2$, four items at depth $3$ and so on. We can regard these as budget restrictions - the adversary cannot choose too many items to map to a particular depth. There are additional combinatorial restrictions, as well 
For example, an adversary cannot place two labels of depth $2$ both to the right of the label of depth $1$ - because even though the root node in a binary search tree can have two children, it cannot have more than one right child.

But suppose we relax this restriction, and only consider budget restrictions on the adversary. Then the resulting best response question becomes a bipartite maximum weight matching problem. Nodes on the left (in this bipartite graph) represent items, and nodes on the right represent depths (there is one node of depth $1$, two nodes of depth $2$, ...). And for any choice of a depth to assign to a node, we can evaluate the value of this decision - if this decision beats median search when searching for that element, we give the corresponding edge weight $1$. If it ties median search, we give the edge weight $\frac{1}{2}$ and otherwise we give the edge zero weight.

We give an upper bound on the value of a maximum weight matching in this graph, hence giving an upper bound on how well an adversary can do if he is subject to only budget restrictions. If we now add the combinatorial restrictions too, this only makes the best response problem harder. So in this way, we are able to bound how much an adversary can beat median search. In fact, we give a lower bound that matches this upper bound - so our relaxation did not make the problem strictly easier (to beat median search).

%
%
%

We focus on the scenario in which $|\Omega| = 2^r - 1$ and $p$ is the uniform
distribution.  Throughout this section we
denote $n = |\Omega|$. The reason we fix $n$ to be of the form $2^r
-1$ is because the optimal single-player strategy is well-defined in
the sense that the first query will be at precisely the median
element, and if the element $\omega$ is not found on this query, then
the problem will break down into one of two possible $2^{r-1} -1$
sized sub-problems.  For this case, we give asymptotically matching upper and lower
bounds on the beatability of median search.

\begin{definition}
We will call a $|\Omega|$-dimensional vector $\vec{u}$ over $\{1, 2,
... |\Omega|\}$ a depth vector (over the universe $\Omega$) if there
is some $T \in X$ such that $\vec{u}_j = c_T(\omega_j)$.
\end{definition}

\begin{proposition}\label{prop:depth}
A $|\Omega|$-dimensional vector $\vec{u}$ over $\{1, 2,
... |\Omega|\}$ is a depth vector (over the universe $\Omega$) if and
only if
\begin{enumerate}
\item exactly one entry of $\vec{u}$ is set to $1$ (let the
  corresponding index be $j$), and
\item the vectors $[\vec{u}_1 -1, \vec{u}_2 -1, .... \vec{u}_{j-1}
  -1]$ and $[\vec{u}_{j+1} -1, \vec{u}_{j+2} -1, .... \vec{u}_{n} -1]$
  are depth vectors over the universe $\{\omega_1, \omega_2,
  ...\omega_{j-1}\}$ and $\{\omega_{j+1}, \omega_{j+2},
  ... \omega_{n}\}$ respectively.
\end{enumerate}
\end{proposition}



\begin{proof}
Given any vector $\vec{u}$ that (recursively) satisfies the above
Conditions $1$ and $2$, one can build up a binary search tree on
$\Omega$ inductively. Let $\omega_j \in \Omega$ be the unique item
such that $\vec{u}_j = 1$ which exists because $\vec{u}$ satisfies
Condition $1$. Since $\vec{u}$ satisfies Condition $2$, the vectors
$\vec{u}_L = [\vec{u}_1 -1, \vec{u}_2 -1, .... \vec{u}_{j-1} -1]$ and
$\vec{u}_R = [\vec{u}_{j+1} -1, \vec{u}_{j+2} -1, .... \vec{u}_{n}
  -1]$ and hence by induction we know that there are binary search
trees $T_L$ and $T_R$ on the universe $\{\omega_1, \omega_2,
...\omega_{j-1}\}$ and $\{\omega_{j+1}, \omega_{j+2},
... \omega_{n}\}$ respectively for which $\vec{u}_L(i) =
c_{T_L}(\omega_i)$ and $\vec{u}_R(i') = c_{T_R}(\omega_{i'})$ for each
$1 \leq i \leq j-1$ and $j + 1 \leq i' \leq n$ respectively.

So we can build a binary search tree $T$ on $\Omega$ by labeling the
root node $\omega_j$ and letting the left sub-tree to $T_L$ and the
right sub-tree to $T_R$. Since the in-order traversal of $T_L$ and of
$T_R$ result in visiting $\{\omega_1, \omega_2, ...\omega_{j-1}\}$ and
$\{\omega_{j+1}, \omega_{j+2}, ... \omega_{n}\}$ in sorted order, the
in-order traversal of $T$ will visit $\Omega$ in sorted order and
hence $T \in X$.

Not also that $c_T(\omega_i) = 1 + c_{T_L}(\omega_i)$ for $1 \leq i
\leq j-1$ and similarly $c_T(\omega_{i'}) = 1 + c_{T_R}(\omega_{i'})$
for $j + 1 \leq i' \leq n$. So this implies that $\vec{u}$ satisfies
$\vec{u}_i = c_T(\omega_i)$ for all $1 \leq i \leq n$, as
desired. This completes the inductive proof that if a vector $\vec{u}$
satisfies Conditions $1$ and $2$, then it is a depth vector.

Conversely, given $T \in X$, there is only one element $\omega_j$ such
that $c_T(\omega_j) =1$ and so Condition $1$ is met. Let $T_L$ and
$T_R$ be the binary search trees that are the left and right sub-tree
of $T$ rooted at $\omega_j$ respectively, where "$\omega_j$" is the
label of the root node in $T$. Again, $c_T(\omega_i) = 1 +
c_{T_L}(\omega_i)$ for $1 \leq i \leq j-1$ and similarly
$c_T(\omega_{i'}) = 1 + c_{T_R}(\omega_{i'})$ for $j + 1 \leq i' \leq
n$ so the vector corresponding to $c_T$ does indeed satisfy Condition
$2$ by induction.
\end{proof}

\begin{claim}~\label{cor:upper}
For any depth vector $\vec{u}$, and any $s \in \{1, 2,
... |\Omega|\}$, $$|\{ j \in [n] | \mbox{ such that } \vec{u}_j = s\}|
\leq 2^{s-1}$$
\end{claim}

\begin{lemma}~\label{lemma:bstbeatlow}
The beatability of median search is at least $\frac{2^{r-1} -1 + 2^{r-3}}{2^r -1} \approx \frac{5}{8}$.
\end{lemma}

\begin{proof}
Consider the depth vector for median search for $2^3 -1$ ($r = 3$): $[3, 2, 3, 1, 3, 2, 3]$ and consider a partially filled vector $[2, 1, *, *, 2, *, *]$.
We can generate the depth vector for median search for $r+1$ from the depth vector for median search for $r$ as follows: alternately interleave values of $r+1$ into the depth vector for $r$. For example the depth vector for median search for $r = 4$ is $[4, 3, 4, 2, 4, 3, 4, 1, 4, 3, 4, 2, 4, 3, 4]$. We assume by induction that all blocks in the partially filled vector are either $*$s or are one less than the corresponding entry in the depth vector for median search. This is true by induction for the base case $r = 3$. We also assume that the $*$s are given in blocks of length exactly two. This is also true in the base case. Then if we consider the depth vector for median search for $r+1$, if an entry of $r+1$ is interleaved, we can place a value of $r$ if the corresponding entry in the partially filled vector is interleaved between two entries that are already assigned numbers. Otherwise three entries are interleaved into a string of exactly two $*$s. The median entry in this string of $5$ symbols corresponds to a newly added $r+1$ entry in the depth vector for median search. At the median of this $5$ symbol string, we can place a value of $r$. This again creates sequences of $*$s of length exactly two, because we have replaced only the median entry in the string of $5$ symbols.

If we are given a partially filled depth vector with the property that one value $1$ is placed, two values of $2$ are placed, four values of $3$ are placed,... and $2^{r-1}$ values of $r$ are placed. Additionally, we require that all unfilled entries (which are given the value $*$ for now) occur in blocks of length exactly $2$. Then we can fill these symbols with the values $r+1$ and $r+2$, such that the value of $r+1$ aligns with a corresponding value of $r+1$ in the depth vector for median search (precisely because any two consecutive symbols contain exactly one value of $r+1$ in the depth vector corresponding to median search for $r+1$).

We can use Proposition~\ref{prop:depth} to prove that this resulting completely filled vector is indeed a depth vector. How much does this strategy beat median search? There are $2^{r} -1$ locations (i.e. every index in which a value of $1$, $2$, ... or $r$ is placed) in which this strategy beats median search. And there are $2^{r-1}$ locations in which this strategy ties median search. Note that this is for $2^{r+1} -1$ items, and so the beatability of median search on $2^r -1$ items is exactly $$\lim_{r \rightarrow \infty} \frac{2^{r-1} -1 + 2^{r-3}}{2^r -1}
= \frac{5}{8}$$
\end{proof}


\begin{lemma}~\label{lemma:bstbeathigh}
The beatability of median search is at most $\frac{2^{r-1} -1 + 2^{r-3}}{2^r -1} \approx \frac{5}{8}$.
\end{lemma}

\begin{proof}
One can give an upper bound on the beatability of median search by
relaxing the question to a matching problem. Given a universe $\Omega$
of size $2^r -1$, consider the following weighted matching problem:
For every value of $s \in \{1, 2, ... r-1\}$, add $2^{s-1}$ nodes on
both the left and right side with label ``s''. For any pair of nodes $a,
b$ where $a$ is contained on the left side, and $b$ is contained on
the right side, set the value of the edge connecting $a$ and $b$ to be
equal to $0$ if the label of $a$ is strictly smaller than the label of
$b$, $\frac{1}{2}$ if the two labels have the same value, and $1$ if
the label of $a$ is strictly larger than the label of $b$.

Let $M$ be the maximum value of a perfect matching. Let $\bar{M}$ be
the average value - i.e. $\frac{M}{2^r -1}$.

\begin{claim}
$\bar{M}$ is an upper bound on the beatability of binary search.
\end{claim}

\begin{proof}
For any $s \in \{1, 2, ... r-1\}$, the depth vector $\vec{u}(M)$
corresponding to median search has exactly $2^{s-1}$ indices $j$ for
which $\vec{u}(M)_j = s$.

We can make an adversary more powerful by allowing the adversary to
choose any vector $\vec{u}$ which satisfies the condition that for any
$s \in \{1, 2, ... |\Omega|\}$, the number of indices $j$ for which
$\vec{u}_j = s$ is at most $2^{s-1}$ because using
Claim~\ref{cor:upper} this is a weaker restriction than requiring the
adversary to choose a vector $\vec{u}$ that is a depth vector. So in
this case, the adversary may as well choose a vector $\vec{u}$ that
satisfies the constraint in Claim~\ref{cor:upper} with equality.

And in this case where we allow the adversary to choose any vector
$\vec{u}$ that satisfies Claim~\ref{cor:upper}, the best response
question is exactly the matching problem described above - because for
each entry in $\vec{u}_M$ because the adversary only needs to choose
what label $s \in \{1, 2, ... r-1\}$ to place at this location subject
to the above budget constraint that at most $2^{s-1}$ labels of type
"s" are used in total.
\end{proof}

\begin{claim}
$\bar{M} \leq \frac{2^{r-1} -1 + 2^{r-3}}{2^r -1}$.
\end{claim}

\begin{proof}
Given a maximum value, bipartite matching problem, the dual covering
problem has variables $y_v$ corresponding to each node $v$, and the
goal is to minimize $\sum_v y_v $ subject to the constraint that for
every edge $(u, v)$ in the graph (which has value $w(u, v)$), the dual
variables satisfy $y_u + y_v \geq w(u, v)$ and each variable $y_v$ is
non-negative.

So we can upper bound $M$ by giving a valid dual solution. This will
then yield an upper bound on $M$ and consequently will also give an
upper bound on $\bar{M}$.

Consider the following dual solution: For each node on the right, with
label "s" for $s < r -2$, set $y_v$ equal to $1$. For a node on the
right with label "s" for $s = r-2$, set $y_v$ equal to $\frac{1}{2}$
and for each label "s" for $s = r-1$, set $y_v = 0$. Additionally, for
every node on the left, only nodes with label "s" for $s = r-1$ are
given non-zero dual variable, and set this variable equal to
$\frac{1}{2}$.

The value of the dual $\sum_v y_v$ is $1 + 2 + ... 2^{r-3} +
\frac{1}{2} 2^{r-2} + \frac{1}{2} 2^{r-1}$. And so this yields an
upper bound on $\bar{M}$ of $\frac{2^{r-1} -1 + 2^{r-3}}{2^r -1}$
and $$\lim_{r \rightarrow \infty} \frac{2^{r-1} -1 + 2^{r-3}}{2^r -1}
= \frac{5}{8}$$

\end{proof}

\end{proof}

\section{Conclusions and Future Directions}
\label{sec:conc}

The dueling framework presents a fresh way of looking at classic optimization
problems through the lens of competition.  As we have demonstrated,
standard algorithms for many optimization problems
do not, in general, perform well in these competitive settings.
This leads us to suspect that alternative algorithms, tailored to competition,
may find use in practice.
We have adapted linear programming and learning techniques into
methods for constructing such algorithms.

We have only just begun an exploration of the dueling framework for
algorithm analysis; there are many open questions yet to consider.
For instance, one avenue of future work is to compare the computational
difficulty of solving an
optimization problem with that of solving the associated duel.  We know
that one is not consistently more difficult than the other:
in Appendix \ref{app:racing} we provide an example in which the
optimization problem is computationally easy but the competitive variant
appears difficult; an example of the opposite situation is given in
Appendix \ref{app:easyduel}, where a computationally hard optimization
problem has a duel which can be solved easily.  Is there some
structure underlying the relationship between the computational
hardness of an optimization problem and its competitive analog?

Perhaps more importantly, one could ask about performance loss inherent
when players choose their algorithms competitively instead of using the
(single-player) optimal algorithm.  In other words, what is the
{\em price of anarchy} \cite{KP99} of a given duel?  Such a question requires a suitable
definition of the social welfare for multiple algorithms, and in particular it may
be that two competing algorithms perform better than a single optimal algorithm.
Our main open question is: \emph{does competition between algorithms
improve or degrade expected performance?}

\pagebreak

\bibliography{dueling}
\bibliographystyle{plain}

\pagebreak

\appendix

\section{Proofs from Section~\ref{sec:defn}}
\label{app:defn}

Here we present the proof of Lemma~\ref{lem:appx}.  The proof follows
a reduction from low-regret learning to computing approximate minmax
strategies \cite{FS96}.  It was shown there that if two players use
``low regret'' algorithms, then the empirical distribution over play
will converge to the set of minmax strategies.  However, instead of
using the weighted majority algorithm, we use the ``Follow the
expected leader'' (FEL) algorithm \cite{KV05}.  That algorithm gives a
reduction between the ability to compute best responses and ``low
regret.''

Note, for this section, we will use the fact that $x^t M x' \in [-C,C]$ for $C=B^3nn'$ under our assumptions on $K,K',$ and $M$.  We will extend the domain of
 $v:\reals_{\geq 0}^n \times \reals_{\geq 0}^{n'} \rightarrow \reals$ naturally by $v(x,x') = x^t M x'$.  For $x \in [0,B]^n$ and $x' \in [0,B]^{n'}$, $v(x,x') \in [-C,C]$.  Additionally, for simplicity we will change the domains of $\cO$ and $\cO'$ to $\reals_{\geq 0}^n$ and $\reals_{\geq 0}^{n'}$, as follows.  For any $x' \in \reals_{\geq 0}^{n'}$, we simply take $\cO(B x' /\|x'\|_\infty)$ as the best response to $x'$ (for $x' = 0$ an arbitrary element of $K$, such as $\cO(0)$ may be chosen).  This scaling is logical since $\arg\max_{x \in K} x^t M x' = \arg\max_{x \in K} x^t M \alpha x'$ for $\alpha>0$.
 By linearity in $v$, it implies that, for the new oracle $\cO$ and any $x' \in \reals_{\geq 0}^{n'}$,
 \begin{equation}\label{eq:scale}
 v(\cO(x'),x') \geq \max_{x \in K} v(x,x') - \eps \frac{\|x'\|_\infty}{B}.
 \end{equation}
 Similarly for $\cO'$.

Fix any sequence length $T\geq 1$. Consider $T$ periods of repeated
play of the duel.  Let the strategies chosen by players 1 and 2, in
period $t$, be $x_t$ and $x'_t$, respectively.  Define the {\em
  regret} of a player 1 on the sequence to be,
$$\max_{x \in K} \sum_{t=1}^T v(x,x'_t) - \sum_{t=1}^T v(x_t,x'_t).$$
Similarly define regret for player 2.  The (possibly negative) regret of a player is how much better that player could have done using the best single strategy, where the best is chosen with the benefit of hindsight.
\begin{observation}\label{ob:1}
Suppose in sequence $x_1,x_2,\ldots,x_T$ and $x'_1,x'_2,\ldots,x'_T$, both players have at most $r$ regret.  Let $\sigma=(x_1+\ldots+x_T)/T$, $\sigma'=(x_1'+\ldots+x_T')/T$ be the uniform mixed strategies over $x_1,\ldots,x_T$, and $x_1',\ldots,x'_T$, respectively.  Then $\sigma$ and $\sigma'$ are $\eps$-minmax strategies, for $\eps=2r/T$.
\end{observation}
\begin{proof}
Say the minmax value of the game is $\alpha$.  Let $a=\frac{1}{T}\sum_t v(x_i,x'_i)$. Then, by the definition of regret, $a \geq \alpha - r/T$, because otherwise player 1 would have more than $r$ regret as seen by any minmax strategy for player 1, which guarantees at least an $\alpha T$ payoff on the sequence.  Also, we have that, against the uniform mixed strategy over $x_1,\ldots,x_T$, no strategy can achieve payoff of at least $a-r$, by the definition of regret (for player 2).  Hence, $\sigma$ guarantees player 1 a payoff of at least $\alpha - 2r/T$.  A similar argument shows that $\sigma'$ is $2r/T$-minmax for player 2.
\end{proof}

The FEL algorithm for a player is
simple.  It has parameters $B,R>0,N \geq 1$ and also takes as input an $\eps$ best response oracle for the
player.  For player 1 with best response orace $\cO$, the algorithm
operates as follows. On each period $t=1,2,\ldots$, it chooses $N$
independent uniformly-random vectors $r_{t1},r_{t2},\ldots,r_{tN} \in
[0,R]^{m'}$.  It plays,
$$\frac{1}{N}\left(\sum_{j=1}^N \cO\left(r_{tj} +
\sum_{\tau=1}^{t-1} x_\tau\right)\right) \in K.$$ The above is seen to
be in $K$ by convexity.  Also recall that for ease of analysis, we
have assumed that $\cO$ takes as input any positive combination of
points in $K'$.

\begin{lemma}\label{lem:a1}
For any $B,C,R,T,\beta,\eps>0$, and any $r \in [0,R]^{m'}$,
$$\sum_{t=1}^T v(\cO(r+x'_1+x'_2+\ldots +x'_t), x'_{t}) \geq \max_{x \in K}
\sum_{t=1}^T v(x, x'_{t}) - 2CR/B - T(T+R/B)\eps.$$
\end{lemma}
The proof is a straightforward modification of Kalai and Vempala's
proof \cite{KV05}.  What this is saying is that the ``be the leader''
algorithm, which is ``one step ahead'' and uses the information for
the current period in choosing the current period's play, has low
regret.  Moreover, one can perturb the payoffs by any amount in a
bounded cube, and this won't affect the bounds significantly.  The
point of the perturbations, which we will choose randomly, will be to
make it harder to predict what the algorithm will do.  For the
analysis, they will make it so that ``be the leader'' and ``follow the
leader'' perform similarly.
\begin{proof}
Define $y_t = r+ x'_1+\ldots+x'_{t-1}$.  We first show,
\begin{equation}\label{eq:woop}
v(\cO(y_1),r) + \sum_{t=1}^T v(\cO(y_{t+1}), x'_{t}) \geq v(\cO(y_{T+1}),r) + \sum_{t=1}^T v(\cO(y_{T+1}), x'_{t}) - T(T+R/B)\eps.
\end{equation}
The facts that $\|r\|_\infty \leq R$ implies that $v(x,r) \in [-CR/B,CR/B]$, and hence,
\begin{align*}
CR/B+\sum_{t=1}^T v(\cO(y_{t+1}), x'_{t}) &\geq \max_{x \in K} \left(v(x,r)+ \sum_{t=1}^T v(x, x'_{t})\right)-T(T+R/B)\eps\\
&\geq \max_{x \in K} \left(\sum_{t=1}^T v(x, x'_{t})\right)-T(T+R/B)\eps-2CR/B,
\end{align*}
which is equivalent to the lemma.
We now prove (\ref{eq:woop}) by induction on $T$.  For $T=0$, we have equality.  For the induction step, it suffices to show that,
$$v(\cO(y_{T}),r) + \sum_{t=1}^{T-1} v(\cO(y_{T}), x'_{t}) \geq  v(\cO(y_{T+1}),r) + \sum_{t=1}^{T-1} v(\cO(y_{T+1}), x'_{t}) - (R/B+T) \eps.$$
However, this is just an inequality between $v(\cO(y_T),y_T)$ and $v(\cO(y_{T+1}),y_T)$, and hence follows from (\ref{eq:scale}) and the fact that $\|y_T\|_\infty/B \leq R/B + T$.  Hence we have established (\ref{eq:woop}) and also the lemma.
\end{proof}

\begin{lemma}\label{lem:a2}
For any $\delta \geq 0$, with probability $\geq 1-2Te^{-2\delta^2 N}$,
$$\sum_{t=1}^T v(x_t,x'_t) \geq \max_{x \in K} \sum_{t=1}^T v(x,x'_t) -\delta C T -2BC m'T/R -2CR/B-T(T+R/B)\eps.$$
\end{lemma}
\begin{proof}
It is clear that $y_t$ and $y_{t+1}$ are similarly distributed.  For any fixed $x_1',x_2',\ldots,x_T'$, define $\bar{x}_t$ by,
$$\bar{x}_t=\frac{1}{R^{m'}}\int_{r \in [0,R]^{m'}}\cO\left(r+x_1'+\ldots+x'_{t-1}\right)dr.$$
By linearity of expectation and $v$, it is easy to see that $\E[x_t|x_1',\ldots,x_{t-1}']=\bar{x}_t$ and, $$\E[v(x_t,x'_t)~|~x_1',\ldots,x_t']=v(\bar{x}_t,x'_t).$$
By Chernoff-Hoeffding bounds, since $v(x_t,x'_t) \in [-C,C]$, for any $\delta\geq 0$, we have that with probability at least $1-e^{-2\delta^2 N}$,
$$\Pr\bigl[~|v(x_t,x'_t) -v(\bar{x}_t,x'_t)| \geq \delta C~\bigl|~x_1',\ldots,x_t'\bigr]\leq 2e^{-2\delta^2 N}.$$
Hence, by the union bound, $\Pr\left[~|\sum_t v(x_t,x'_t) - \sum_t v(\bar{x}_t,x'_t)|\geq \delta CT\right] \leq 2Te^{-2\delta^2 N}.$

The key observation of Kalai and Vempala is that $\bar{x}_t$ and $\bar{x}_{t+1}$ are close because the $m'$-dimensional translated cubes $x_1'+\ldots+x'_{t-1}+[0,R]^{m'}$ and $x_1'+\ldots+x'_{t}+[0,R]^{m'}$ overlap significantly.  In particular, they overlap in on all but at most a $B m'/R$ fraction \cite{KV05} of their volume.  Since $v$ is in $[-1,1]$, this means that $\bigl|v(\bar{x}_t,x'_t)-v(\bar{x}_{t+1},x'_t)\bigr| \leq 2BC m'/R$.  This follows from the fact that $v$ is bilinear, and hence when moved into the integral has exactly the same behavior on all but a $B m'/R$ fraction of the points in each cube.  This implies, that with probability $\geq 1-2Te^{-2\delta^2N}$,
$$\sum_{t=1}^T v(x_t,x'_t) \geq \sum_{t=1}^T v(\bar{x}_{t+1},x'_t) -\delta CT -2BC m'T/R.$$
Combining this with Lemma \ref{lem:a1} completes the proof.
\end{proof}

We are now ready to prove Lemma \ref{lem:appx}.
\begin{proof}[Proof of Lemma \ref{lem:appx}]
We take $T=\left(4C \sqrt{\max(m,m')}/(3\eps)\right)^{2/3}$, $R = B\sqrt{\max(m,m') T}$ and $N=\ln(4TC/\delta)/(2\eps^2)$. As long as $T \geq \max(m,m')$, $R/B \leq T$ and hence Lemma \ref{lem:a2} implies that with probability at least $1-\delta$, if both players play FEL then both will have regret at most $$\eps T +4 C \sqrt{\max(m,m') T} +2T^2\eps \leq 4 C \sqrt{\max(m,m') T} + 3T^2\eps \leq 12(\max(m,m')C^2)^{2/3}\eps^{-1/3}.$$  Observation \ref{ob:1} completes the proof.
\end{proof}

\section{A Racing Duel}
\label{app:racing}

The racing duel illustrates a simple example in which the beatability is unbounded, the optimization problem is ``easy,'' but finding polynomial-time minmax algorithms remains a challenging open problem.  The optimization problem behind the racing duel is routing under uncertainty.
There is an underlying directed multigraph $(V,E)$ containing designated start and terminal nodes $s,t \in V$, along with a distribution over bounded weight vectors $\Omega \subset \reals_{\geq 0}^{E}$, where $\omega_{e}$ represents the delay in traversing edge $e$.  The feasible set $X$ is the set of paths from $s$ to $t$.  The probability distribution $p \in \Delta(\Omega)$ is an arbitrary measure over $\Omega$.  Finally, $c(x,\omega) = \sum_{e \in x} \omega_e$.

For general graphs, solving the racing duel seems quite challenging.  
This is true even when routing between two nodes with parallel edges, i.e., $V=\{s,t\}$ and all edges $E=\{e_1,e_2,\ldots,e_n\}$ are from $s$ to $t$.  As mentioned in the introduction, this problem is in some sense a ``primal'' duel in the sense that it can encode any duel and finite strategy set.  In particular, given any optimization problem with $|X|=n$, we can create a race where each edge $e_i \in E$ corresponds to a strategy $x_i \in X$, and the delays on the edges match the costs of the associated strategies.

\subsection{Shortest path routing is 1-beatable}

The single-player racing problem is easy: take the shortest path on the graph with weights $w_e=\E_{\omega\sim p}[\omega_e]$.  However, this algorithm can be beaten almost always.  Consider a graph with two parallel edges, $a$ and $b$, both from $s$ to $t$.  Say the cost of $a$ is $\epsilon/2>0$ with probability 1, and the cost of $b$ is 0 with probability $1-\epsilon$ and $1$ with probability $\epsilon$.  The optimization algorithm will choose $a$, but $b$ beats $a$ with probability $1-\epsilon$, which is arbitrarily close to 1.

\subsection{Price of anarchy}

Take social welfare to be the average performance, $W(x,x') = (c(x)+c(x'))/2$.  Then the price of anarchy for racing is unbounded.  Consider a graph with two parallel edges, $a$ and $b$, both from $s$ to $t$.  The cost of $a$ is $\epsilon>0$ with probability 1, and the cost of $b$ is 0 with probability 3/4 and $1$ with probability 1/4.  Then $b$ a dominant strategy for both players, but its expected cost is $1/4$, so the price of anarchy is $1/(4\eps)$, which can be arbitrarily large.

\section{When Competing is Easier than Playing Alone}
\label{app:easyduel}

Recall that the racing problem from Appendix \ref{app:racing} was
``easy'' for single-player optimization, yet seemingly difficult to
solve in the competitive setting.  We now give a contrasting example:
a problem for which
competing is easier than solving the single-player optimization.

The intuition behind our construction is as follows.  The optimization problem will be based upon a computationally difficult decision problem, which an algorithm must attempt to answer.  After the algorithm submits an answer, nature provides its own ``answer'' chosen uniformly at random.  If the algorithm disagrees with nature, it incurs a large cost that is independent of whether or not it was correct.  If the algorithm and nature agree, then the cost of answering the problem correctly is less than the cost of answering incorrectly.

More formally, let $L \subseteq \{0,1\}^*$ be an arbitrary language, and let $z \in \{0,1\}^*$ be a string.  Our duel will be $D(X,\Omega,p,c)$ where $X = \Omega = \{0,1\}$, $p$ is uniform, and the cost function is
\[
c(x,\omega) =
\begin{cases}
0 & \text{if $(x=\omega=1$ and $z \in L)$ or $(x=\omega=0$ and $z \not\in L)$} \\
1 & \text{if $(x=\omega=1$ and $z \not\in L)$ or $(x=\omega=0$ and $z \in L)$} \\
2 & \text{if $x\neq\omega$} \\
\end{cases}
\]
The unique optimal solution to this (single-player) problem is to output $1$ if and only if $z \in L$.
Doing so is as computationally difficult as the decision problem itself.
On the other hand, finding a minmax optimal algorithm is trivial for every $z$ and $L$, since \emph{every} algorithm has value $1/2$: for any $x'$, $v(1-x',x') = \Pr[ \omega \neq x' ] = 1/2 = v(x',x')$.

\section{Asymmetric Games}
\label{app:asymmetric}

We note that all of the examples we considered
have been symmetric with respect to the players, but our results
can be extended to asymmetric games.  Our analysis of bilinear duels
in Section \ref{sec:bilinear} does not assume symmetry when discussing bilinear
games.  For instance, we could consider a game where player 1 wins in
the case of ties, so player 1's payoff is $\Pr[c(x,\omega)\leq c(x',\omega)]$.
One natural example would be a ranking duel in which there is an
``incumbent'' search engine that appeared first, so a user prefers to
continue using it rather than switching to a new one.
This game can be represented in the same bilinear form as
in Section \ref{sec:rank}, the only change being a small modification
of the payoff matrix $M$.  Other types of asymmetry, such as players
having different objective functions, can be handled in the same way.
For example, in a hiring duel, our analysis techniques apply even if
the two players may have different pools of candidates, of possibly
different sizes and qualities.

\end{document}